\documentclass[a4paper,UKenglish, numberwithinsect]{lipics-v2019}

\title{Two-way Parikh Automata}

\author{Emmanuel~Filiot\footnote{E.~Filiot is a research associate of F.R.S.-FNRS.
		He is supported by the ARC Project
		Transform F\'ed\'eration Wallonie-Bruxelles and the
		FNRS CDR project J013116F.}
}{Universit\'e libre de Bruxelles }{}{}{}

\author{Shibashis~Guha\footnote{S.~Guha is supported by the ARC project “Non-Zero Sum Game Graphs: Applications to Reactive Synthesis and Beyond” ( F\'ed\'eration Wallonie-Bruxelles)}
}{Universit\'e libre de Bruxelles}{}{}{}

\author{Nicolas~Mazzocchi \footnote{N~Mazzocchi is a PhD student funded by a FRIA fellowship from the F.R.S.-FNRS.}
}{Universit\'e libre de Bruxelles}{}{}{}

\authorrunning{E.~Filiot and S.~Guha and N.~Mazzocchi}

\Copyright{Emmanuel~Filiot and Shibashis~Guha and Nicolas~Mazzocchi}


\nolinenumbers

\usepackage{bbold} 
\usepackage{mathtools} 
\usepackage{xspace}
\usepackage{amsmath, amssymb, amsfonts, amsthm}
\usepackage{wasysym}
\usepackage{tikz}
\usetikzlibrary{automata, calc}

\newcommand{\N}{{\mathbb{N}}}		
\newcommand{\Z}{{\mathbb{Z}}}		
\newcommand{\Q}{\mathbb{Q}}			

\newcommand{\poly}{\ensuremath{\mathsf{poly}}}
\newcommand{\SIGMA}{\ensuremath{\mathsf{\Sigma}}}
\newcommand{\PI}{\ensuremath{\mathsf{\Pi}}}

\newcommand{\PTime}{\textsc{P}\xspace}
\newcommand{\NPTime}{\textsc{NP}\xspace}
\newcommand{\NPTimeC}{\textsc{NP-C}\xspace}
\newcommand{\coNPTime}{\textsc{coNP}\xspace}

\newcommand{\ExpTime}{\textsc{Exp}\xspace}
\newcommand{\NExpTime}{\textsc{NExp}\xspace}
\newcommand{\NExpTimeC}{\textsc{NExp-C}\xspace}
\newcommand{\coNExpTime}{\textsc{coNExp}\xspace}
\newcommand{\coNExpTimeC}{\textsc{coNExp-C}\xspace}
\newcommand{\NPSpace}{\textsc{NPSpace}\xspace}
\newcommand{\PSpace}{\textsc{PSpace}\xspace}
\newcommand{\PSpaceC}{\textsc{PSpace-C}\xspace}
\newcommand{\ExpSpace}{\textsc{ExpSpace}\xspace}

\newcommand{\PA}{{\textsf{PA}}}
\newcommand{\DFA}{{\textsf{DFA}}}
\newcommand{\UFA}{{\textsf{UFA}}}
\newcommand{\NFA}{{\textsf{FA}}}
\newcommand{\DPA}{{\textsf{DPA}}}
\newcommand{\UPA}{{\textsf{UPA}}}
\newcommand{\NPA}{{\textsf{PA}}}
\newcommand{\twoPA}{{\textsf{2PA}}}
\newcommand{\twoDFA}{{\textsf{2DFA}}}
\newcommand{\twoUFA}{{\textsf{2UFA}}}
\newcommand{\twoNFA}{{\textsf{2FA}}}
\newcommand{\twoDPA}{{\textsf{2DPA}}}

\newcommand{\twoNPA}{{\textsf{2PA}}}

\newcommand{\lword}{\ensuremath{{\vdash}}}
\newcommand{\rword}{\ensuremath{{\dashv}}}

\renewcommand{\P}{\mathfrak{P}} 
\newcommand{\proj}{\ensuremath{\text{proj}}} 
\newcommand{\range}{\ensuremath{\text{range}}} 
\renewcommand{\O}{\mathcal{O}} 

\newcommand{\new}{\mathrel{\raisebox{-.5ex}{\scalebox{0.9}{$\stackrel{\text{def}}{=}$}}}}

\newcommand{\C}{\mathcal{C}} 
\newcommand{\onleft}{\textsf{L}}
\newcommand{\onright}{\textsf{R}}

\begin{document}

\maketitle

	\begin{abstract}
          Parikh automata extend automata with counters whose values can only
          be tested at the end of the computation, with respect to
          membership into a semi-linear set. Parikh automata have
          found several applications, for instance in transducer
          theory, as they enjoy decidable emptiness problem.

          In this paper, we study two-way Parikh automata. We
          show that emptiness becomes undecidable in the
          non-deterministic case. However, it is \PSpaceC when the number of
          visits to any input position is bounded and
          the semi-linear set is given as an existential Presburger
          formula. We also give tight complexity bounds for the inclusion,
          equivalence and universality problems. Finally, we
          characterise precisely the complexity of those problems when
          the semi-linear constraint is given by an arbitrary
          Presburger formula. 
	\end{abstract}

	\newcommand{\figureScale}{.72}

\tikzset{
	pics/splited/.style n args={3}{
		code = {
			\node[state] (#1) {};
			\draw[dotted] (#1.45) -- (#1.225);
			\node[xshift = -4.4pt, yshift=4.6pt] at (#1.center) {#3};
			\node[xshift = 4pt, yshift=-4pt] at (#1.center) {#2};
		}
	},
	pics/initial splited/.style n args={3}{
		code = {
			\node[state, initial, initial text=] (#1) {};
			\draw[dotted] (#1.45) -- (#1.225);
			\node[xshift = -4pt, yshift=4pt] at (#1.center) {#3};
			\node[xshift = 4pt, yshift=-4pt] at (#1.center) {#2};
		}
	},
	pics/binitial splited/.style n args={3}{
		code = {
			\node[state, initial above, initial text=] (#1) {};
			\draw[dotted] (#1.45) -- (#1.225);
			\node[xshift = -4pt, yshift=4pt] at (#1.center) {#3};
			\node[xshift = 4pt, yshift=-4pt] at (#1.center) {#2};
		}
	},
	pics/accepting splited/.style n args={3}{
		code = {
			\node[state, accepting] (#1) {};
			\draw[dotted] (#1.45) -- (#1.225);
			\node[xshift = -4pt, yshift=4pt] at (#1.center) {#3};
			\node[xshift = 4pt, yshift=-4pt] at (#1.center) {#2};
		}
	},
	pics/final splited/.style n args={3}{
		code = {
			\tikzstyle{accepting}=[accepting by arrow,accepting text=,accepting where= right]
			\node[state, accepting] (#1) {};
			\draw[dotted] (#1.45) -- (#1.225);
			\node[xshift = -4pt, yshift=4pt] at (#1.center) {#3};
			\node[xshift = 4pt, yshift=-4pt] at (#1.center) {#2};
		}
	}
}


\newcommand{\drawUniversality}{
	\scalebox{\figureScale}{
		\begin{tikzpicture}[>=stealth, node distance=2cm, thick]
		\node[state, initial right, initial text=] (qi) {$\onright$};
		\node[state, below of = qi] (q) {$\onright$};
		\node[state, accepting, below of = q] (qf) {$\onright$};
		\node[xshift=1.3cm] at (qf) {$\Psi(x)$};
		
		\path[->]
			(qi) edge node[right] {$\lword \mid 0$} (q)
			(q) edge[loop right] node[right] {$a \mid 1$} (q)
			(q) edge node[right] {$\rword \mid 0$} (qf)
		;
		\end{tikzpicture}
	}
}


\newcommand{\drawTimesGeneralised}{
	\scalebox{\figureScale}{
		\begin{tikzpicture}[>=stealth, node distance=2cm, thick]
			\node[state, initial, initial text=] (qi) {$\onleft$};
			\node[state, right of = qi] (q1) {$\onleft$};
			\node[state, below of = q1] (p1) {$\onright$};
			
			\node[draw=white, state, right of = q1] (q) {$\dots$};
			
			\node[state, right of = q] (qn) {$\onleft$};
			\node[state, below of = qn] (pn) {$\onright$};
			\node[state, accepting, right of = qn] (qf) {$\onleft$};

			\path[->]
				(qi) edge node[above] {$\lword$} (q1)
				(q1) edge[loop above] node[above] {$a$} (q1)
				(q1) edge[bend left] node[right] {$\# \mid k_1$++} (p1)
				(p1) edge[loop below] node[below] {$a \mid c_1$++} (p1)
				(p1) edge[bend left] node[left] {$\lword$} (q1)
				
				(q1) edge node[above] {$\#$} (q)
				(q) edge node[above] {$\#$} (qn)
				
				(qn) edge[loop above] node[above] {$a$} (qn)
				(qn) edge[bend left] node[right] {$\rword \mid k_n$++} (pn)
				(pn) edge[loop below] node[below] {$a \mid c_n$++} (pn)
				(pn) edge[bend left] node[left] {$\#$} (qn)
				(qn) edge node[above] {$\rword$} (qf)
			;
		\end{tikzpicture}
	}
}


\newcommand{\drawRight}{
	\scalebox{\figureScale}{
		\begin{tikzpicture}[>=stealth, node distance=3cm, thick]
		\node[state, initial below, initial text=$T_\onright$] (qi) {};
		\node[state, accepting, right of = qi, right =0cm] (qf) {};
		\node[state, left of= qi] (q) {};

		\path[->]
			(qi) edge[loop above] node[above] {$(pp)_{p \in \Q^\onleft} \mid \varepsilon$} (qi)
			(qi) edge[bend left] node[below] {$(q)_{q \in Q^\onright} \mid q$} (q)
			(q) edge node[above] {$(p)_{p \in Q^\onleft} \mid p$} (qi)
			(qi) edge node[above] {$(q)_{q \in Q^\onright} \mid q$} (qf)
		;
		\end{tikzpicture}
	}
}


\newcommand{\drawLeft}{
	\scalebox{\figureScale}{
		\begin{tikzpicture}[>=stealth, node distance=3cm, thick]
		\node[state, initial below, initial text=$T_\onleft$] (qi) {};
		\node[state, accepting, right of = qi, right =0cm] (qf) {};
		\node[state, right of= qf] (q) {};

		\path[->]
			(qi) edge node[above] {$(q)_{q \in Q^\onright} \mid q$} (qf)
			(qf) edge node[above] {$(p)_{p \in Q^\onleft} \mid p$} (q)
			(q) edge[bend left] node[below] {$(q)_{q \in Q^\onright} \mid q$} (qf)
			(qf) edge[loop above] node[above] {$(qq)_{q \in \Q^\onright} \mid \varepsilon$} (q)
		;
		\end{tikzpicture}
	}
}


\newcommand{\drawMatch}{
	\scalebox{\figureScale}{
		\begin{tikzpicture}[>=stealth, node distance=6mm, very thick]
		\node (a0) {};
		\node[right of = a0, right = 6mm] (b0) {};
		\node[right of = b0, right = 6mm] (c0) {};
		\node[right of = c0, right = 6mm] (d0) {};
		\node[right of = d0, right = 6mm] (e0) {};
		\node[right of = e0, right = 6mm] (f0) {};
		\node[right of = f0, right = 6mm] (g0) {};
		\node[right of = g0, right = 6mm] (h0) {};
		
		\node[below of = a0, below = -5mm] (a1) {$\onright$};
		\node[below of = b0, below = -5mm] (b1) {$\onright$};
		\node[below of = c0, below = -5mm] (c1) {$\onright$};
		\node[below of = d0, below = -5mm] (d1) {$\onleft$};
		\node[below of = e0, below = -3.9mm] (e1) {};
		\node[below of = f0, below = -3.9mm] (f1) {};
		\node[below of = g0, below = -3.9mm] (g1) {};
		\node[below of = h0, below = -3.9mm] (h1) {};
		
		\node[below of = a1] (a2) {};
		\node[below of = b1] (b2) {};
		\node[below of = c1] (c2) {$\onright$};
		\node[below of = d1] (d2) {$\onleft$};
		\node[below of = e1] (e2) {};
		\node[below of = f1] (f2) {};
		\node[below of = g1] (g2) {};
		\node[below of = h1] (h2) {};
		
		\node[below of = a2] (a3) {};
		\node[below of = b2] (b3) {};
		\node[below of = c2] (c3) {$\onright$};
		\node[below of = d2] (d3) {$\onright$};
		\node[below of = e2] (e3) {$\onright$};
		\node[below of = f2] (f3) {$\onright$};
		\node[below of = g2] (g3) {$\onright$};
		\node[below of = h2] (h3) {$\onleft$};
		
		\node[below of = a3] (a4) {};
		\node[below of = b3] (b4) {};
		\node[below of = c3] (c4) {};
		\node[below of = d3] (d4) {};
		\node[below of = e3] (e4) {$\onright$};
		\node[below of = f3] (f4) {$\onleft$};
		\node[below of = g3] (g4) {$\onleft$};
		\node[below of = h3] (h4) {$\onleft$};
		
		\node[below of = a4] (a5) {};
		\node[below of = b4] (b5) {};
		\node[below of = c4] (c5) {};
		\node[below of = d4] (d5) {};
		\node[below of = e4] (e5) {$\onright$};
		\node[below of = f4] (f5) {$\onleft$};
		\node[below of = g4] (g5) {};
		\node[below of = h4] (h5) {};
		
		\node[below of = a5] (a6) {};
		\node[below of = b5] (b6) {};
		\node[below of = c5] (c6) {$\onright$};
		\node[below of = d5] (d6) {$\onleft$};
		\node[below of = e5] (e6) {$\onleft$};
		\node[below of = f5] (f6) {$\onleft$};
		\node[below of = g5] (g6) {};
		\node[below of = h5] (h6) {};
		
		\node[below of = a6] (a7) {};
		\node[below of = b6] (b7) {};
		\node[below of = c6] (c7) {$\onright$};
		\node[below of = d6] (d7) {$\onright$};
		\node[below of = e6] (e7) {$\onright$};
		\node[below of = f6] (f7) {$\onright$};
		\node[below of = g6] (g7) {$\onright$};
		\node[below of = h6] (h7) {$\onright$};

		\path
		(a0) edge[draw=none] node[above] {$\lword$} (b0)
		(c0) edge[draw=none] node[above] {$a$} (d0)
		(e0) edge[draw=none] node[above] {$b$} (f0)
		(g0) edge[draw=none] node[above] {$\rword$} (h0)
		
		(a1) edge[->] (b1)		(c1) edge[->] (d1)
		(c2) edge[<-] (d2)
		(c3) edge[->] (d3)		(e3) edge[->] (f3)		(g3) edge[->] (h3)
		(e4) edge[<-] (f4)	  (g4) edge[<-] (h4)
		(e5) edge[->] (f5)
		(c6) edge[<-] (d6)		(e6) edge[<-] (f6)
		(c7) edge[->] (d7)		(e7) edge[->] (f7)		(g7) edge[->] (h7)
		
		(b1) edge[dashed] (c1)
		(d3) edge[dashed](e3)
		(f3) edge[dashed] (g3)
		(f4) edge[dashed](g4)
		(d6) edge[dashed] (e6)
		(d7) edge[dashed] (e7)
		(f7) edge[dashed] (g7)
		
		(d1.east) edge[thick, dotted, bend left] (d2.east)
		(h3.east) edge[thick, dotted, bend left] (h4.east)
		(f5.east) edge[thick, dotted, bend left] (f6.east)
		
		(c2.west) edge[thick, dotted, bend right] (c3.west)
		(e4.west) edge[thick, dotted, bend right] (e5.west)
		(c6.west) edge[thick, dotted, bend right] (c7.west)
		;
		\end{tikzpicture}
	}
}
\newcommand{\drawMatchBis}{
	\scalebox{\figureScale}{
		\begin{tikzpicture}[>=stealth, node distance=6mm, very thick]
		
		\foreach \x in {1, 2, 3, 4, 5, 6, 7, 8}{
			\foreach \y in {1, 2, 3, 4, 5, 6, 7, 8}{
				\node[fill=white, circle,minimum size=6mm] (q\x\y) at (\x*16mm, -\y*6mm) {};
			}
		}
	
		\foreach \v in {12, 22, 32, 33, 34, 44, 54, 64, 74, 55, 56, 37, 38, 48, 58, 68, 78, 88}{
			\node at (q\v) {$\onright$};
		}
		\foreach \v in {42, 43, 84, 85, 75, 65, 66, 67, 57, 47}{
			\node at (q\v) {$\onleft$};
		}
	
		\path
			(q11) edge[draw=none] node[above] {$\lword$} (q21)
			(q31) edge[draw=none] node[above] {$a$} (q41)
			(q51) edge[draw=none] node[above] {$b$} (q61)
			(q71) edge[draw=none] node[above] {$\rword$} (q81)
			
			(q12) edge[->] (q22)		(q32) edge[->] (q42)
													(q33) edge[<-] (q43)
													(q34) edge[->] (q44)		(q54) edge[->] (q64)		(q74) edge[->] (q84)
																							(q55) edge[<-] (q65)		(q75) edge[<-] (q85)
																							(q56) edge[->] (q66)
													(q37) edge[<-] (q47)		(q57) edge[<-] (q67)
													(q38) edge[->] (q48)		(q58) edge[->] (q68)		(q78) edge[->] (q88)
													
			(q22) edge[dashed] (q32)
			(q44) edge[dashed](q54)
			(q64) edge[dashed] (q74)
			(q65) edge[dashed](q75)
			(q47) edge[dashed] (q57)
			(q48) edge[dashed] (q58)
			(q68) edge[dashed] (q78)
			
			(q42.east) edge[thick, dotted, bend left] (q43.east)
			(q84.east) edge[thick, dotted, bend left] (q85.east)
			(q66.east) edge[thick, dotted, bend left] (q67.east)
			
			(q33.west) edge[thick, dotted, bend right] (q34.west)
			(q55.west) edge[thick, dotted, bend right] (q56.west)
			(q37.west) edge[thick, dotted, bend right] (q38.west)
		;
		\end{tikzpicture}
	}
}


\newcommand{\drawSection}{
	\scalebox{\figureScale}{
		
		\begin{tikzpicture}[>=stealth, node distance=6mm, very thick]
			\node (a0) {};
			\node[right of = a0, right = 6mm] (b0) {};
			\node[right of = b0, right = 6mm] (c0) {};
			\node[right of = c0, right = 6mm] (d0) {};
			\node[right of = d0, right = 6mm] (e0) {};
		
			\node[below of = a0, below = -6.1mm] (a1) {$q_i$};
			\node[below of = b0, below = -5mm] (b1) {};
			\node[below of = c0, below = -5mm] (c1) {};
			\node[below of = d0, below = -5mm] (d1) {};
			\node[below of = e0, below = -5mm] (e1) {};
			
			\node[below of = a1] (a2) {};
			\node[below of = b1] (b2) {};
			\node[below of = c1] (c2) {};
			\node[below of = d1] (d2) {};
			\node[below of = e1] (e2) {};
			
			\node[below of = a2] (a3) {};
			\node[below of = b2] (b3) {};
			\node[below of = c2] (c3) {};
			\node[below of = d2] (d3) {};
			\node[below of = e2] (e3) {};
			
			\node[below of = a3] (a4) {};
			\node[below of = b3] (b4) {};
			\node[below of = c3] (c4) {};
			\node[below of = d3] (d4) {};
			\node[below of = e3] (e4) {};
			
			\node[below of = a4] (a5) {};
			\node[below of = b4] (b5) {};
			\node[below of = c4] (c5) {};
			\node[below of = d4] (d5) {};
			\node[below of = e4] (e5) {};
			
			\node[below of = a5] (a6) {};
			\node[below of = b5] (b6) {};
			\node[below of = c5] (c6) {};
			\node[below of = d5] (d6) {};
			\node[below of = e5] (e6) {};
			
			\node[below of = a6] (a7) {};
			\node[below of = b6] (b7) {};
			\node[below of = c6] (c7) {};
			\node[below of = d6] (d7) {};
			\node[below of = e6] (e7) {$q_f$};

			\path
				(a0) edge[draw=none] node[above] {$\lword$} (b0)
				(b0) edge[draw=none] node[above] {$a$} (c0)
				(c0) edge[draw=none] node[above] {$b$} (d0)
				(d0) edge[draw=none] node[above] {$\rword$} (e0)
				
				(a1) edge[->] (b1)		(b1) edge[->] (c1)
												(b2) edge[<-] (c2)
												(b3) edge[->] (c3)		(c3) edge[->] (d3)		(d3) edge[->] (e3)
																				(c4) edge[<-] (d4)	  (d4) edge[<-] (e4)
																				(c5) edge[->] (d5)
												(b6) edge[<-] (c6)		(c6) edge[<-] (d6)
												(b7) edge[->] (c7)		(c7) edge[->] (d7)		(d7) edge[->] (e7)
				
				(c1.center) edge[bend left] (c2.center)
				(b2.center) edge[bend right] (b3.center)
				(e3.center) edge[bend left] (e4.center)
				(c4.center) edge[bend right] (c5.center)
				(d5.center) edge[bend left] (d6.center)
				(b6.center) edge[bend right] (b7.center)
			;
		\end{tikzpicture}
	}
}


\newcommand{\drawTimes}{
	\scalebox{\figureScale}{
		\begin{tikzpicture}[>=stealth, node distance=2.8cm, thick]
		\pic {binitial splited={q_0}{$\onright$}{$q_0$}};s
		\pic[right of = q_0] {splited={q_1}{$\onright$}{$q_1$}};
		\pic[right of = q_1] {splited={q_2}{$\onright$}{$q_2$}};
		\pic[right of = q_2] {splited={q_3}{$\onright$}{$q_3$}};
		\pic[right of = q_3] {accepting splited={q_4}{$\onright$}{$q_4$}};
		\pic[left of = q_0] {splited={q_5}{$\onleft$}{$q_5$}};
		\node [yshift = -1cm] at (q_4) {$S = \{ (0, 0) \} $};

		\path[->]
			(q_0) edge node[above] {$\lword \mid (0, 0)$} (q_1)
			(q_1) edge[loop below] node[below] {$a \mid (0, 0)$} (q_1)
			(q_1) edge node[above] {$\# \mid (0, 0)$} (q_2)
			(q_2) edge[loop below] node[below] {$a \mid (0, -1)$} (q_2)
			(q_2) edge node[above] {$\# \mid (0, 0)$} (q_3)
			(q_3) edge[loop below] node[below] {$a \mid (-1, 0)$} (q_3)
			(q_3) edge node[above] {$\rword \mid (0, 0)$} (q_4)
			
			(q_1) edge[bend left] node[below] {$\# \mid (0, 1)$} (q_5)
			(q_5) edge[loop left] node[left] {$\begin{array}{l|l}
				a & (1, 0) \smallskip\\ \# & (0, 0) \end{array}$} (q_5)
			(q_5) edge node[above] {$\lword\mid (0, 0)$} (q_0)
		;
		\end{tikzpicture}
	}
}


\newcommand{\drawLn}{
	\scalebox{\figureScale}{
		\begin{tikzpicture}[>=stealth, node distance=2.6cm, thick]
		\pic {initial splited={A}{$\onright$}{$q_I$}};
		\pic[right of = A] {splited={B}{$\onright$}{$q^a$}};
		\pic[below right of = B] {splited={C0}{$\onright$}{$q_0$}};
		
		\pic[above right of = C0] {splited={Cb1}{$\onright$}{$q_1^b$}};
		\node[circle, minimum size = 1cm, right of = Cb1] (Cb) {};
		\pic[right of = Cb] {splited={Cbn}{$\onright$}{$q_n^b$}};
		
		\pic[below right of = C0] {splited={Cc1}{$\onright$}{$q_1^c$}};
		\node[circle, minimum size = 1cm, right of = Cc1] (Cc) {};
		\pic[right of = Cc] {splited={Ccn}{$\onright$}{$q_n^c$}};
		
		\pic at ($(Cbn)!0.5!(Ccn)$) {splited={D1}{$\onleft$}{$p_1$}};
		\node[circle, minimum size = 1cm] (D) at ($(Cb)!0.5!(Cc)$) {};
		\pic at ($(Cb1)!0.5!(Cc1)$) {splited={Dn}{$\onleft$}{$p_n$}};

		\node[below left of = C0] (E1) {};
		\node[left of = E1] (E2) {};
		\node(E) at ($(E1)!0.5!(E2)$) {};
		\pic at (E) {accepting splited={F}{$\onright$}{$q_F$}};
		\node[xshift = -15mm] at (F) {$S = \{ 0 \} $};
		
		\draw[dotted]
		(Cb.east) -- (Cb.west)
		(Cc.east) -- (Cc.west)
		(D.east) -- (D.west)
		;
		
		\node[xshift=2cm, yshift=0cm] (down) at (F) {};
		\node[xshift=1.9319cm, yshift=0.51764cm] (middle) at (F) {};
		\node[xshift=1.7321cm, yshift=1cm] (up) at (F) {};

		\path[->]
		(up) edge[dashed] node[above left] {$\rword \mid 0$} (F)
		(middle) edge[dashed] (F)
		(down) edge[dashed] node[below] {$\rword \mid 0$} (F)
		
		(A) edge node[above] {$\lword \mid 0$} (B)
		(B) edge[loop right] node[right] {$a \mid 1$} (B)
		(B) edge node[below left] {$\# \mid 0$} (C0)
		
		(C0) edge node[above left] {$b \mid 0$} (Cb1)
		(Cb1) edge node[above] {$b, c \mid 0$} (Cb)
		(Cb) edge node[above] {$b, c \mid 0$} (Cbn)
		(Cbn) edge node[right] {$\begin{array}{l|l}
			b & 0  \smallskip\\ c & -1 \end{array}$} (D1)
		
		(C0) edge node[below left] {$c \mid 0$} (Cc1)
		(Cc1) edge node[above] {$b, c \mid 0$} (Cc)
		(Cc) edge node[above] {$b, c \mid 0$} (Ccn)
		(Ccn) edge node[right] {$\begin{array}{l|l}
			b & -1  \smallskip\\ c & 0 \end{array}$} (D1)
		
		(D1) edge node[above] {$b, c \mid 0$} (D)
		(D) edge node[above] {$b, c \mid 0$} (Dn)
		(Dn) edge node[above] {$b, c \mid 0$} (C0)
		;
		\end{tikzpicture}
	}
}


\newcommand{\drawLnOld}{
	\scalebox{\figureScale}{
		\begin{tikzpicture}[>=stealth, node distance=2.8cm, thick]
		\pic {initial splited={A}{$\onright$}{$q_0$}};
		\pic[right of = A] {splited={B}{$\onright$}{$q_1$}};
		\pic[right of = B] {splited={C}{$\onright$}{$q_2$}};
		\node[state, rounded corners, rectangle, above right of = C] (Cb) {$\begin{smallmatrix}
			\onright \smallskip\\ n-1\end{smallmatrix}$};
		\node[state, rounded corners, rectangle, below right of = C] (Cc) {$\begin{smallmatrix}
			\onright \smallskip\\ n-1\end{smallmatrix}$};
		\node[state, rounded corners, rectangle, above right of = Cc] (D) {$\begin{smallmatrix}
			\onleft \smallskip\\ n-1 \end{smallmatrix}$};
		\pic[below left of = C] {accepting splited={E}{$\onright$}{$q_f$}};
		\node[xshift = -15mm] at (E) {$S = \{ 0 \} $};
		
		\draw[dotted]
			(Cb.east) -- (Cb.west)
			(Cc.east) -- (Cc.west)
			(D.east) -- (D.west)
		;
		
		\path[->]
			(A) edge node[above] {$\lword \mid 0$} (B)
			(B) edge[loop above] node[above] {$a \mid 1$} (B)
			(B) edge node[above] {$\# \mid 0$} (C)
			(C) edge node[above left] {$b \mid 0$} (Cb)
			(C) edge node[below left] {$c \mid 0$} (Cc)
			(Cb) edge node[above right] {$\begin{array}{l|l}
					b & 0  \smallskip\\ c & -1 \end{array}$} (D)
			(Cc) edge node[below right] {$\begin{array}{l|l}
				b & -1  \smallskip\\ c & 0 \end{array}$} (D)
			(D) edge node[above] {$b, c \mid 0$} (C)
			(C) edge node[below right] {$\rword \mid 0$} (E)
		;
		\end{tikzpicture}
	}
}

\newcommand{\drawGadget}{
	\scalebox{\figureScale}{
		\begin{tikzpicture}[>=stealth, node distance=2.8cm, thick]
			\node[state, rounded corners, rectangle] (r) {$\begin{smallmatrix}
				X \in \{ \onleft, \onright\} \smallskip\\ x=n-1 \end{smallmatrix}$};
			\pic[right of = r] {initial splited={q_0}{$X$}{$p_0$}};
			\pic[right of = q_0] { splited={q_1}{$X$}{$p_1$}};
			\node[right of = q_1] (q) {\dots};
			\pic[right of = q] {final splited={q_n}{$X$}{$p_x$}};
			\node [xshift = 1.3cm] at (r)  {$\new$};
			
			\draw[dotted] (r.east) -- (r.west);
			
			\path[->]
				(q_0) edge node[above] {$b, c \mid 0$} (q_1)
				(q_1) edge node[above] (m) {$b, c \mid 0$} (q)
				(q) edge node[above] {$b, c \mid 0$} (q_n)
			;
			
			\node[yshift = -2cm] (p) at (m) {};
			\pic at (p) {accepting splited={f}{$\onright$}{$q_f$}};
			
			\path[->]
				(q_0) edge node[below left] {$\rword \mid 0$} (f)
				(q_1) edge[pos=.7] node[above right] {$\rword \mid 0$} (f)
				(q_n) edge node[below right] {$\rword \mid 0$} (f)
			; 
			
		\end{tikzpicture}
	}
}


\newcommand{\drawSectionTOP}{
	\scalebox{\figureScale}{
		\newcommand{\myleftskip}{25mm}
		\newcommand{\mybelowskip}{-8mm}
		\begin{tikzpicture}[>=stealth, node distance=9mm, very thick]
		\node (a0) at (0, 0) {};
		\node[xshift = \myleftskip] (b0) at (a0) {};
		\node[xshift = \myleftskip] (c0) at (b0) {};
		\node[xshift = \myleftskip] (d0) at (c0) {};
		\node[xshift = \myleftskip] (e0) at (d0) {};
		
		\node[yshift = \mybelowskip] (a1) at (a0) {$q_1$};
		\node[yshift = \mybelowskip] (b1) at (b0) {$q_2$};
		\node[yshift = \mybelowskip] (c1) at (c0) {};
		\node[yshift = \mybelowskip] (d1) at (d0) {};
		\node[yshift = \mybelowskip] (e1) at (e0) {};
		
		\node[below of = a1] (a2) {};
		\node[below of = b1] (b2) {};
		\node[below of = c1] (c2) {};
		\node[below of = d1] (d2) {};
		\node[below of = e1] (e2) {};
		
		\node[below of = a2] (a3) {};
		\node[below of = b2] (b3) {};
		\node[below of = c2] (c3) {$q_5$};
		\node[below of = d2] (d3) {$q_6$};
		\node[below of = e2] (e3) {};
		
		\node[below of = a3] (a4) {};
		\node[below of = b3] (b4) {};
		\node[below of = c3] (c4) {};
		\node[below of = d3] (d4) {$q_8$};
		\node[below of = e3] (e4) {};
		
		\node[below of = a4] (a5) {};
		\node[below of = b4] (b5) {};
		\node[below of = c4] (c5) {};
		\node[below of = d4] (d5) {};
		\node[below of = e4] (e5) {};
		
		\node[below of = a5] (a6) {};
		\node[below of = b5] (b6) {};
		\node[below of = c5] (c6) {$q_{11}$};
		\node[below of = d5] (d6) {};
		\node[below of = e5] (e6) {};
		
		\node[below of = a6] (a7) {};
		\node[below of = b6] (b7) {};
		\node[below of = c6] (c7) {$q_{13}$};
		\node[below of = d6] (d7) {$q_{14}$};
		\node[below of = e6] (e7) {$q_{15}$};

		\path
			(a0) edge[draw=none] node[above] {$\lword$} (b0)
			(b0) edge[draw=none] node[above] (qn) {$a$} (c0)
			(c0) edge[draw=none] node[above] {$b$} (d0)
			(d0) edge[draw=none] node[above] {$\rword$} (e0)
			
			(a1) edge[->] node[above] {$\vec{v}_1$} (b1)		(b1) edge[->] node[above] {$\vec{v}_2$} (c1)
			(b2) edge[<-] node[above] {$\vec{v}_3$} (c2)
			(b3) edge[->] node[above] {$\vec{v}_4$} (c3)		(c3) edge[->] node[above] {$\vec{v}_5$}(d3)		(d3) edge[->] node[above] {$\vec{v}_6$} (e3)
			(c4) edge[<-] node[above] {$\vec{v}_7$} (d4)	  (d4) edge[<-] node[above] {$\vec{v}_8$} (e4)
			(c5) edge[->] node[above] {$\vec{v}_9$} (d5)
			(b6) edge[<-] node[above] {$\vec{v}_{10}$} (c6)		(c6) edge[<-] node[above] {$\vec{v}_{11}$} (d6)
			(b7) edge[->] node[above] {$\vec{v}_{12}$} (c7)		(c7) edge[->] node[above] {$\vec{v}_{13}$} (d7)		(d7) edge[->] node[above] {$\vec{v}_{14}$} (e7)
			
			(c1.center) edge[bend left] node[right] (qe) {$q_3$} (c2.center)
			(b2.center) edge[bend right] node[left] {$q_4$} (b3.center)
			(e3.center) edge[bend left] node[right] {$q_7$} (e4.center)
			(c4.center) edge[bend right] node[left] {$q_9$} (c5.center)
			(d5.center) edge[bend left] node[right] {$q_{10}$} (d6.center)
			(b6.center) edge[bend right] node[left] (qw) {$q_{12}$} (b7.center)
		;
		
		\draw[dashed, gray]
			(qw.west |- qn.north) rectangle (c7.south -| qe.east)
		;
		
		\end{tikzpicture}
	}
}


\newcommand{\drawSectionBis}{
	\scalebox{\figureScale}{
		\newcommand{\myleftskip}{40mm}
		\newcommand{\mybelowskip}{-10mm}
		\begin{tikzpicture}[>=stealth, node distance=9mm, very thick]
		\node (a0) at (0, 0) {};
		\node[xshift = .3*\myleftskip] (aa0) at (a0) {};
		\node[xshift = .7*\myleftskip] (aaa0) at (a0) {};
		
		\node[xshift = \myleftskip] (b0) at (a0) {};
		\node[xshift = .3*\myleftskip] (bb0) at (b0) {};
		\node[xshift = .7*\myleftskip] (bbb0) at (b0) {};
		
		\node[xshift = \myleftskip] (c0) at (b0) {};
		\node[xshift = .3*\myleftskip] (cc0) at (c0) {};
		\node[xshift = .7*\myleftskip] (ccc0) at (c0) {};
		
		\node[xshift = \myleftskip] (d0) at (c0) {};
		\node[xshift = .3*\myleftskip] (dd0) at (d0) {};
		\node[xshift = .7*\myleftskip] (ddd0) at (d0) {};
		
		\node[xshift = \myleftskip] (e0) at (d0) {};
		\node[xshift = .3*\myleftskip] (ee0) at (e0) {};
		\node[xshift = .7*\myleftskip] (eee0) at (e0) {};

		\node[yshift = \mybelowskip] (a1) at (a0) {$q_1$};
		\node[xshift = .3*\myleftskip] (aa1) at (a1) {};
		\node[xshift = .7*\myleftskip] (aaa1) at (a1) {};
		
		\node[yshift = \mybelowskip] (b1) at (b0) {$q_2$};
		\node[xshift = .3*\myleftskip] (bb1) at (b1) {};
		\node[xshift = .7*\myleftskip] (bbb1) at (b1) {};
		
		\node[yshift = \mybelowskip] (c1) at (c0) {};
		\node[xshift = .3*\myleftskip] (cc1) at (c1) {};
		\node[xshift = .7*\myleftskip] (ccc1) at (c1) {};
		
		\node[yshift = \mybelowskip] (d1) at (d0) {};
		\node[xshift = .3*\myleftskip] (dd1) at (d1) {};
		\node[xshift = .7*\myleftskip] (ddd1) at (d1) {};
		
		\node[yshift = \mybelowskip] (e1) at (e0) {};
		\node[xshift = .3*\myleftskip] (ee1) at (e1) {};
		\node[xshift = .7*\myleftskip] (eee1) at (e1) {};

		\node[below of = a1] (a2) {};
		\node[xshift = .3*\myleftskip] (aa2) at (a2) {};
		\node[xshift = .7*\myleftskip] (aaa2) at (a2) {};
		
		\node[below of = b1] (b2) {};
		\node[xshift = .3*\myleftskip] (bb2) at (b2) {};
		\node[xshift = .7*\myleftskip] (bbb2) at (b2) {};
		
		\node[below of = c1] (c2) {};
		\node[xshift = .3*\myleftskip] (cc2) at (c2) {};
		\node[xshift = .7*\myleftskip] (ccc2) at (c2) {};
		
		\node[below of = d1] (d2) {};
		\node[xshift = .3*\myleftskip] (dd2) at (d2) {};
		\node[xshift = .7*\myleftskip] (ddd2) at (d2) {};
		
		\node[below of = e1] (e2) {};
		\node[xshift = .3*\myleftskip] (ee2) at (e2) {};
		\node[xshift = .7*\myleftskip] (eee2) at (e2) {};

		\node[below of = a2] (a3) {};
		\node[xshift = .3*\myleftskip] (aa3) at (a3) {};
		\node[xshift = .7*\myleftskip] (aaa3) at (a3) {};
		
		\node[below of = b2] (b3) {};
		\node[xshift = .3*\myleftskip] (bb3) at (b3) {};
		\node[xshift = .7*\myleftskip] (bbb3) at (b3) {};
		
		\node[below of = c2] (c3) {$q_5$};
		\node[xshift = .3*\myleftskip] (cc3) at (c3) {};
		\node[xshift = .7*\myleftskip] (ccc3) at (c3) {};
		
		\node[below of = d2] (d3) {$q_6$};
		\node[xshift = .3*\myleftskip] (dd3) at (d3) {};
		\node[xshift = .7*\myleftskip] (ddd3) at (d3) {};
		
		\node[below of = e2] (e3) {};
		\node[xshift = .3*\myleftskip] (ee3) at (e3) {};
		\node[xshift = .7*\myleftskip] (eee3) at (e3) {};

		\node[below of = a3] (a4) {};
		\node[xshift = .3*\myleftskip] (aa4) at (a4) {};
		\node[xshift = .7*\myleftskip] (aaa4) at (a4) {};
		
		\node[below of = b3] (b4) {};
		\node[xshift = .3*\myleftskip] (bb4) at (b4) {};
		\node[xshift = .7*\myleftskip] (bbb4) at (b4) {};
		
		\node[below of = c3] (c4) {};
		\node[xshift = .3*\myleftskip] (cc4) at (c4) {};
		\node[xshift = .7*\myleftskip] (ccc4) at (c4) {};
		
		\node[below of = d3] (d4) {$q_8$};
		\node[xshift = .3*\myleftskip] (dd4) at (d4) {};
		\node[xshift = .7*\myleftskip] (ddd4) at (d4) {};
		
		\node[below of = e3] (e4) {};
		\node[xshift = .3*\myleftskip] (ee4) at (e4) {};
		\node[xshift = .7*\myleftskip] (eee4) at (e4) {};

		\node[below of = a4] (a5) {};
		\node[xshift = .3*\myleftskip] (aa5) at (a5) {};
		\node[xshift = .7*\myleftskip] (aaa5) at (a5) {};
		
		\node[below of = b4] (b5) {};
		\node[xshift = .3*\myleftskip] (bb5) at (b5) {};
		\node[xshift = .7*\myleftskip] (bbb5) at (b5) {};
		
		\node[below of = c4] (c5) {};
		\node[xshift = .3*\myleftskip] (cc5) at (c5) {};
		\node[xshift = .7*\myleftskip] (ccc5) at (c5) {};
		
		\node[below of = d4] (d5) {};
		\node[xshift = .3*\myleftskip] (dd5) at (d5) {};
		\node[xshift = .7*\myleftskip] (ddd5) at (d5) {};
		
		\node[below of = e4] (e5) {};
		\node[xshift = .3*\myleftskip] (ee5) at (e5) {};
		\node[xshift = .7*\myleftskip] (eee5) at (e5) {};

		\node[below of = a5] (a6) {};
		\node[xshift = .3*\myleftskip] (aa6) at (a6) {};
		\node[xshift = .7*\myleftskip] (aaa6) at (a6) {};
		
		\node[below of = b5] (b6) {};
		\node[xshift = .3*\myleftskip] (bb6) at (b6) {};
		\node[xshift = .7*\myleftskip] (bbb6) at (b6) {};
		
		\node[below of = c5] (c6) {$q_{11}$};
		\node[xshift = .3*\myleftskip] (cc6) at (c6) {};
		\node[xshift = .7*\myleftskip] (ccc6) at (c6) {};
		
		\node[below of = d5] (d6) {};
		\node[xshift = .3*\myleftskip] (dd6) at (d6) {};
		\node[xshift = .7*\myleftskip] (ddd6) at (d6) {};
		
		\node[below of = e5] (e6) {};
		\node[xshift = .3*\myleftskip] (ee6) at (e6) {};
		\node[xshift = .7*\myleftskip] (eee6) at (e6) {};

		\node[below of = a6] (a7) {};
		\node[xshift = .3*\myleftskip] (aa7) at (a7) {};
		\node[xshift = .7*\myleftskip] (aaa7) at (a7) {};
		
		\node[below of = b6] (b7) {};
		\node[xshift = .3*\myleftskip] (bb7) at (b7) {};
		\node[xshift = .7*\myleftskip] (bbb7) at (b7) {};
		
		\node[below of = c6] (c7) {$q_{13}$};
		\node[xshift = .3*\myleftskip] (cc7) at (c7) {};
		\node[xshift = .7*\myleftskip] (ccc7) at (c7) {};
		
		\node[below of = d6] (d7) {$q_{14}$};
		\node[xshift = .3*\myleftskip] (dd7) at (d7) {};
		\node[xshift = .7*\myleftskip] (ddd7) at (d7) {};
		
		\node[below of = e6] (e7) {$q_{15}$};
		\node[xshift = .3*\myleftskip] (ee7) at (e7) {};
		\node[xshift = .7*\myleftskip] (eee7) at (e7) {};

		\path
		(a0) edge[loosely dotted] node[fill=white] {$\lword$} (b0)
		(b0) edge[loosely dotted] node[fill=white] {$a$} (c0)
		(c0) edge[loosely dotted] node[fill=white] {$b$} (d0)
		(d0) edge[loosely dotted] node[fill=white] {$\rword$} (e0)
		
		(a1) edge[->] node[above] {$\vec{v}_1$} (b1)		(b1) edge[->] node[above] (top) {$\vec{v}_2$} (bbb1)
		(bb2) edge[<-] node[above] {$\vec{v}_3$} (bbb2)
		(bb3) edge[->] node[above] {$\vec{v}_4$} (c3)		(c3) edge[->] node[above] {$\vec{v}_5$}(d3)		(d3) edge[->] node[above] {$\vec{v}_6$} (ddd3)
		(cc4) edge[<-] node[above] {$\vec{v}_8$} (d4)	  (d4) edge[<-] node[above] {$\vec{v}_7$} (ddd4)
		(cc5) edge[->] node[above] {$\vec{v}_9$} (ccc5)
		(bb6) edge[<-] node[above] {$\vec{v}_{11}$} (c6)		(c6) edge[<-] node[above] {$\vec{v}_{10}$} (ccc6)
		(bb7) edge[->] node[above] {$\vec{v}_{12}$} (c7)		(c7) edge[->] node[above] {$\vec{v}_{13}$} (d7)		(d7) edge[->] node[above] {$\vec{v}_{14}$} (e7)
		
		(bbb1.center) edge[bend left] node[right] {$q_3$} (bbb2.center)
		(bb2.center) edge[bend right] node[left] {$q_4$} (bb3.center)
		(ddd3.center) edge[bend left] node[right] {$q_7$} (ddd4.center)
		(cc4.center) edge[bend right] node[left] {$q_9$} (cc5.center)
		(ccc5.center) edge[bend left] node[right] {$q_{10}$} (ccc6.center)
		(bb6.center) edge[bend right] node[left] {$q_{12}$} (bb7.center)
		;
		
		\draw[loosely dashed, gray, rounded corners=7]
		(b1.west |- top.north) rectangle (c7.south -| c3.east)
		;
		
		\end{tikzpicture}
	}
}

	\section{Introduction} \label{sec:introduction}
		
\emph{Parikh automata}, introduced in~\cite{KlaRue03}, extend finite automata with counters in $\Z$ which can be incremented and decremented, but the counters can only be tested at the end of the computation, for membership in a semi-linear set (represented for instance as an existential Presburger formula).
More precisely, transitions are of the form $(q, \sigma, \vec{v}, q')$ where $q,q'$ are states, $\sigma$ is an input symbol and $\vec{v} \in \Z^d$ is a vector of dimension $d$.
A word $w$ is accepted if there exists a run $\rho$ on $w$ reaching an
accepting state and whose final vector (the component-wise sum of all
vectors along $\rho$) belongs to a given semi-linear set. Parikh automata strictly extend the expressive power of finite automata.
For example, the context-free language of words of the form $a^nb^n$ is definable by a deterministic Parikh automaton which checks membership in $a^*b^*$, counts the number of occurrences of $a$ and $b$, and at the end tests for equality of the counters, i.e.\ membership in the linear set $\{ (n,n)\mid n\in\N\}$.
They still enjoy decidable, \NPTimeC, non-emptiness problem~\cite{DBLP:conf/lics/FigueiraL15}.

Parikh automata ($\PA$) have found applications for instance in \emph{transducer} theory, in particular to the equivalence problem of functional transducers on words, and to check structural properties of transducers~\cite{patternLogic}, as well as in answering queries in graph databases~\cite{DBLP:conf/lics/FigueiraL15}.
Extensions of Parikh automata with a pushdown stack have been considered in~\cite{PPDA} with positive decidability results with respect to emptiness.
Two-way Parikh automata with a visibly pushdown stack have been considered in~\cite{DBLP:conf/fossacs/DartoisFT19} with applications to tree transducers.

In this paper, our objective is to study \emph{two-way Parikh automata} ($\twoPA$), the extension of $\PA$ with a two-way input head, where the semi-linear set is given by an existential Presburger formula.
For $\twoPA$ as well as subclasses such as deterministic $\twoPA$ ($\twoDPA$), we aim at characterizing the precise complexity of their decision problems (membership, emptiness, inclusion, equivalence), and analysing their expressiveness and closure properties.

\subparagraph{Contributions}
Since semi-linear sets are closed under all Boolean operations, it is easily seen that deterministic Parikh automata ($\DPA$) are closed under all Boolean operations.
More interestingly, it is also known that, while they strictly extend the expressive power of $\DPA$, \emph{unambiguous} $\PA$ ($\UPA$) are (non-trivially) closed under complement (as well as union and intersection)~\cite{cadilhac}.
We give here a simple explanation to these good closure properties: $\UPA$ \emph{effectively} correspond to $\twoDPA$.
Closure of $\twoDPA$ under Boolean operations indeed holds straightforwardly due to determinism.
The conversion of $\UPA$ to $\twoDPA$ is however non-trivial, but is obtained by the very same result on word transducers: it is known that unambiguous finite transducers are equivalent to two-way deterministic finite transducers~\cite{Chytil}, based on a construction by Aho, Hopcroft and Ullman~\cite{Aho1969}, recently improved by one exponential in~\cite{DBLP:conf/icalp/DartoisFJL17}.
Parikh automata can be seen as transducers producing sequences of vectors (the vectors occurring on their transitions), hence yielding the result.
The conversion of $\twoDPA$ to $\UPA$ is a standard construction based on \emph{crossing sections}, which however needs to be carefully analysed for complexity purposes.

The effective equivalence between $\twoDPA$ and $\UPA$ indeed entails decidability of the non-emptiness problem for $\twoDPA$.
However, given that non-emptiness of $\PA$ is known to be \NPTimeC~\cite{DBLP:conf/lics/FigueiraL15}, and the conversion of $\twoDPA$ to $\UPA$ is exponential, this leads to \NExpTime complexity.
By a careful analysis of this conversion and small witnesses properties of Presburger formulas, we show that emptiness of $\twoDPA$, and even \emph{bounded-visit} $\twoNPA$, is actually \PSpaceC.
 Bounded-visit $\twoNPA$ are non-deterministic $\twoNPA$ such that for some natural number $k$, each position of an input word $w$ is visited at most $k$ times by any accepting computation on $w$.
 In particular, $\twoDPA$ are always $n$-visit for $n$ the number of states.
 If the number $k$ of visits is a fixed constant, non-emptiness is then \NPTimeC, which entails complexity result of~\cite{DBLP:conf/lics/FigueiraL15} for (one-way) $\NPA$ (by taking $k=1$).
 We show that dropping the bounded-visit restriction however leads to undecidability.

Thanks to the closure properties of $\twoDPA$, we show that the inclusion, universality and equivalence problems are all \coNExpTimeC.
Those problems are known be undecidable for $\NPA$~\cite{KlaRue03}.
The membership problem of $\twoNPA$ turns out to be \NPTimeC, just as for (one-way) $\NPA$.
The \coNExpTime lower bound holds for one-way deterministic Parikh automata, a result which is also new, to the best of our knowledge.

Finally, we study the extension of two-way Parikh automata with a
semi-linear set defined by a $\SIGMA_i$-Presburger formula, i.e.\ a
formula with a \emph{fixed} number $i$ of unbounded blocks of
quantifiers where the consecutive blocks alternate $i{-}1$ times between existential and universal blocks, and the first block is existential.
We characterise tightly the complexity of the non-emptiness problem for bounded-visit $\SIGMA_i$-$\twoNPA$, as well as the universality, inclusion and equivalence problems for $\SIGMA_i$-$\twoDPA$, in the weak exponential hierarchy~\cite{Haase14}.
For $i>1$, we find that the complexity of these problems is dominated
by the complexity of checking satisfiability or validity of $\SIGMA_i$-Presburger formulas.
This is unlike the case $i=1$: the non-emptiness problem for bounded-visit $\twoNPA$ is \PSpaceC while satisfiability of $\SIGMA_1$-formulas is \NPTimeC.

\subparagraph*{Related work}
Parikh automata are known to be equivalent to reversal-bounded multicounter machines (\textsf{RBCM}) \cite{DBLP:journals/jacm/Ibarra78} in the sense that they describe the same class of languages~\cite{cadilhac}.
Two-way \textsf{RBCM} (\textsf{2RBCM}), even deterministic, are known to have undecidable emptiness problem~\cite{DBLP:journals/jacm/Ibarra78}.
While, using diophantine equations as in the case of \cite{DBLP:journals/jacm/Ibarra78}, we show that emptiness of $\twoNPA$ is undecidable, our decidability result for $\twoDPA$ contrasts with the undecidabilty of deterministic \textsf{2RBCM}.
The difference is that \textsf{2RBCM} can test their counters at any moment during a computation, and not only at the end.
Based on the fact that the number of reversals is bounded, deferring the tests at the end of the computation is always possible~\cite{DBLP:journals/jacm/Ibarra78} but non-determinism is needed.
Unlike $\twoDPA$, deterministic \textsf{2RBCM} are not necessarily bounded-visit.
A $\twoDPA$ can be seen as a deterministic \textsf{2RBCM} whose tests on counters are only done at the end of a computation.

Two-way Parikh automata on \emph{nested words} have been studied in~\cite{DBLP:conf/fossacs/DartoisFT19} where it is shown that under the \emph{single-use} restriction (a generalisation of the bounded-visit restriction to nested words), they have \NExpTimeC non-emptiness problem.
Bounded-visit $\twoNPA$ are a particular case of those Parikh automata operating on (non-nested) words.
Applying the result of~\cite{DBLP:conf/fossacs/DartoisFT19} to $\twoNPA$ would yield a non-optimal \NExpTime complexity for the non-emptiness problem, as it first goes through an explicit but exponential transformation into a one-way machine with known \NPTimeC non-emptiness problem.
Here instead, we rely on a small witness property, whose proof uses a transformation into one-way Parikh automaton, and then we apply a \PSpace algorithm performing on-the-fly the one-way transformation up to some bounded length.

Finally, the emptiness problem for the intersection of $n$ $\NPA$ was shown to be \PSpaceC in~\cite{DBLP:conf/lics/FigueiraL15}.
Our \PSpaceC result on $\twoNPA$ emptiness generalises this result, as the intersection of $n$ $\NPA$ can be simulated trivially by a (sweeping) $n$-bounded $\twoNPA$.
The main lines of our proof are similar to those in~\cite{DBLP:conf/lics/FigueiraL15}, but in addition, it needs a one-way transformation on top of the proof in~\cite{DBLP:conf/lics/FigueiraL15}, and a careful analysis of its complexity.


	\section{Two-way Parikh automata} \label{sec:model}
		Two-way Parikh automata are two-way automata extended with weight vectors and a semi-linear acceptance condition.
In this section, we first define two-way automata, semi-linear sets and then two-way Parikh automata.

\subparagraph{Two-way Automata}
A \emph{two-way finite automaton} ($\twoNFA$ for short) $A$ over an alphabet $\Sigma$ is a tuple $(Q, Q_I, Q_H, Q_F, \Delta)$ whose components are defined as follows.
We let $\lword$ and $\rword$ be two delimiters not in $\Sigma$, intended to represent the beginning and the end of the word respectively.
The set $Q$ is a non-empty finite set of states partitioned into the set of right-reading states $Q^\onright$ and the set of left-reading states $Q^\onleft$.
Then, $Q_I \subseteq Q^\onright$ is the set of initial states, $Q_H\subseteq Q$ is the set of halting states, and $Q_F \subseteq Q_H$ is the set of accepting states.
The states belonging to $Q_H \setminus Q_F$ are said to be \emph{rejecting}.
Finally, $\Delta \subseteq Q \times (\Sigma\cup\{ \lword,\rword\}) \times Q$ is the set of transitions.
 Intuitively, the reading head of $A$ is always placed in between input positions, a transition from $q \in Q^\onright$ (resp.\ $q \in Q^\onleft$) reads the input letter on the right (resp.\ left) of the head and moves the head one step to the right (resp.\ left).
 We also have the following restrictions on the behaviour of the head to keep it in between the boundaries $\lword$ and $\rword$ and to ensure the following properties on the initial and the halting states.
\begin{enumerate}
	\item no outgoing transition from a halting state:\\
	$(Q_H \times \left( \Sigma \cup \{ \lword, \rword \}\right) \times Q) \cap \Delta = \varnothing$
	\item the head cannot move left (resp.\ right) when it is to the
          left of $\lword$ (resp.\ right of $\rword$):\\
	$(Q^\onleft \times \{ \lword \} \times Q^\onleft) \cap \Delta
        = \varnothing$ (resp.\ $(Q^\onright \times \{ \rword \} \times ( Q^\onright \setminus Q_F)) \cap \Delta = \varnothing$)
	\item all transitions leading to a halting state $q_H$ read the delimiter $\rword$:\\
	$((q, a, q_H) \in \Delta \land q_H \in Q_H) \implies (q \in Q^\onright \land a = \rword)$
\end{enumerate}

A configuration $(u^\onleft, p, u^\onright) $ of $A$ on a word $u\in\Sigma^*$ consists of a state $p$ and two words $u^\onleft, u^\onright\in(\Sigma\cup \{\lword,\rword\})^*$ such that $u^\onleft u^\onright=\lword u \rword$.
A \emph{run} $\rho$ on a word $u \in \Sigma^*$ is a sequence $\rho = (u_0^\onleft, q_0, u_0^\onright)a_1(u_1^\onleft, q_1, u_1^\onright) \dots a_n (u_n^\onleft, q_n, u_n^\onright)$ alternating between configurations on $u$ and letters in $\Sigma \cup \{\lword,\rword\}$ such that for all $1 \leq i \leq n$, we have $(q_{i-1}, a_{i}, q_{i}) \in \Delta$, and for all $s\in\{ \onleft, \onright \}$, if $q_{i-1} \in Q^s$ then $|u_i^{s}| = |u_{i-1}^{s}|-1$.
The length of the run $\rho$, denoted $|\rho|$ is the number of letters appearing in $\rho$.
Here $|\rho| = n$.
The \emph{run} $\rho$ is \emph{halting} if $q_n\in Q_H$ (and hence $u_n^\onright=\varepsilon$ by condition~3), \emph{initial} if $u_0^\onleft = \varepsilon$ and $q_0\in Q_I$, \emph{accepting} if it is both initial and halting, and $q_n\in Q_F$; otherwise the run is \emph{rejecting}.
A word $u$ is accepted by $A$ if there exists an accepting run of $A$ on $\lword u \rword$, and the language $L(A)$ of $A$ is defined as the set of words it accepts.


An automaton $A$ is said to be \emph{one-way} (\NFA) if $Q^\onleft$ is empty.
A run $\rho$ is said to be \emph{$k$-visit} if every input position is visited at most $k$ times in the run $\rho$, i.e.\ for $\rho = (u_0^\onleft, q_0, u_0^\onright)\dots (u_n^\onleft, q_n, u_n^\onright)$, we have $\max \{ |P|\mid P\subseteq \{0,\dots,n\}\wedge \forall i,j\in P,\ u_i^\onleft = u_j^\onleft\} \leq k$.
$A$ is said to be $k$-visit if all its accepting runs are $k$-visit, and \emph{bounded-visit} if it is $k$-visit for some $k$.
Also, $A$ is said to be \emph{deterministic} if for all $p \in Q$ and all $a \in \Sigma\cup \{\lword, \rword\}$ there exists at most one $q \in Q$ such that $(p, a, q) \in \Delta$.
Finally, it is \emph{unambiguous} (denoted by the class $\twoUFA$ or $\UFA$ depending on whether it is two-way or one-way) if for every input word there exists at most one accepting run.
The following proposition is trivial but useful:
\begin{proposition} \label{prop:visit}
	Any bounded-visit $\twoNFA$ with $n$ states is $k$-visit for some $k \le n$.
\end{proposition}


\subparagraph{Semi-linear Sets}
Let $d \in \N_{\neq 0}$.
A set $L \subseteq \Z^d$ of dimension $d$ is \emph{linear} if there exist $\vec{v}_0, \dots, \vec{v}_k \in \Z^d$ such that $L = \{ \vec{v}_0+\sum_{i=1}^k x_i \vec{v}_i \mid x_1, \dots, x_n\in \N\}$.
The vectors $(\vec{v}_i)_{1\leq i\leq k}$ are the \emph{periods} and $\vec{v}_0$ is called the \emph{base}, forming what we call a \emph{period-base representation} of $L$, whose size is $d \cdot (k+1) \cdot \log_2(\mu + 1 )$ where $\mu$ is the maximal absolute integer appearing on the vectors.
A set is \emph{semi-linear} if it is a finite union of linear sets.
A period-base representation of a semi-linear set is given by a period-base representation for each of the linear sets it is composed of, and its size is the sum of the sizes of all those representations.


Alternatively, a semi-linear set of dimension $d$ can be represented as the models of a Presburger formula with $d$ free variables.
A Presburger formula is a first-order formula built over terms $t$ on the signature $\{0, 1, +, \times_2\} \cup X$, where $X$ is a countable set of variables and $\times_2$ denotes the doubling (unary) function\footnote{The function $\times_2$ is syntactic sugar allowing us to have simpler binary encoding of values}.
In particular, Presburger formulas obey the following syntax:
$$
	\Phi \new t \leq t \mid \exists x\ \Phi \mid \Phi \land \Phi \mid \Phi \lor \Phi \mid \lnot \Phi
$$
The class of formulas of the form $\exists x_1, \forall x_2 \dots, \Omega_{i} x_{i} \left[\varphi\right]$ where $\varphi$ is quantifier free and $\Omega \in \{ \forall, \exists \}$ is denoted by $\SIGMA_i$.
In particular, $\SIGMA_1$ is the set of existential Presburger formulas.
The size $|\Psi|$ of a formula is its number of symbols.
We denote by $\vec{v}\models \varphi$ the fact that a vector $\vec{v}$
of dimension $d$ satisfies a formula $\varphi$ with $d$ free
variables, and that $\varphi$ is satisfiable is there exists such
$\vec{v}$. We say that $\varphi$ is valid if it is satisfied by any
$\vec{v}$. It is well-known~\cite{semilinear} that a set $S\subseteq \Z^d$ is semi-linear iff there exists an existential Presburger formula $\psi$ with $d$ free variables such that $S = \{ \vec{v}\mid \vec{v} \models \psi\}$.

Let $\Sigma = \{a_1, \dots, a_n\}$ be an alphabet (assumed to be ordered), and $u \in \Sigma^*$, the Parikh image of $u$ is defined as the vector $\P(u) = (|u|_{a_1}, \dots, |u|_{a_n})$ where $|u|_a$ denotes the number of times $a$ occurs in $u$.
The Parikh image of  language $L \subseteq \Sigma^*$ is $\P(L) = \{\P(u) | u \in L\}$.
Parikh's theorem states that the Parikh image of any context-free language is semi-linear.


\subparagraph{Two-way Parikh automata}
A two-way Parikh automaton ($\twoNPA$) of dimension $d\in \N$ over $\Sigma$ is a tuple $P = (A, \lambda, \psi)$ where $A = (Q, Q_I, Q_H, Q_F, \Delta)$ is a \twoNFA{} over $\Sigma$, $\lambda \colon \Delta \rightarrow \Z^d$ maps transitions to vectors, and $\psi$ is an \emph{existential} Presburger formula with $d$ free variables, and is called the \emph{acceptance constraint}.
The \emph{value} $V(\rho)$ of a run $\rho$ of $A$ is the sum of the vectors occurring on its transitions, with $V(\rho) = 0_{\Z^d}$ if $|\rho|=0$.
A word is accepted by $P$ if it is accepted by some accepting run $\rho$ of $A$ and $V(\rho)\models \psi$.
The language $L(P)$ of $P$ is the of words it accepts.
The automaton $P$ is said to be \emph{one-way}, \emph{two-way}, \emph{$k$-visit}, \emph{unambiguous} and \emph{deterministic} if its underlying automaton $A$ is so.
We define the representation size\footnote{Note that weight vectors are not memorized on transition but into a table and transition only carry a key of this table to refer the corresponding weight vectors} of $P$ as $|P| = |Q| + |\psi| + |\range(\lambda)|\big(d \log_2(\mu+1) + |Q|^2\big)$ where $\range(\lambda) = \{ \lambda(t) \mid t\in \Delta \}$ and $\mu$ is the maximal absolute entries appearing in weight vectors of $P$.
Finally two $\twoNPA$ are \emph{equivalent} if they accept the same language.


\subparagraph{Examples}
Let $\Sigma = \{ a, b, c, \#\}$ and for all $n\in\N$, let $L_n = \{
a^k\# u\mid u\in\{ b, c \}^*\land k = |\{ i\mid 1\leq i\leq
|u|-n\wedge u[i]\neq u[i+n]\}|\}$, i.e.\ $k$ is the number of
positions $i$ in $u$ such that the $i$th letter
$u[i]$ mismatches with $u[i+n]$.
For all $n$, $L_n$ is accepted by the $\twoDPA$ of
Fig.~\ref{fig:mismatch} which has $O(n)$ states, tagged with
$\textsf{R}$ or $\textsf{L}$ to indicate whether they are right- or
left-reading respectively. %
On a word $w$, the automaton starts by reading $a^k$ and increments
its counter to store the value $k$ (state $q^a$). 
Then, for the first $|u|-n$ positions $i$ of $u$, the automaton checks
whether $u[i]\neq u[i+n]$ in which case the counter is decremented. To
do so, it stores $\sigma = u[i]$ in its state, moves $n+1$ times to the
right (states $q_0,q_1^\sigma,\dots,q_n^\sigma$), checks whether
$u[i+n]\neq u[i]$ (transitions $q_n^\sigma$ to $p_1$) and decrements the
counter accordingly. Then, it moves $n$ times to the left (states
$p_1$ to $p_n$). Whenever it reads $\dashv$ from states $q_j^\sigma$,
$p_j$ or $q_0$, it moves to state $q_F$ and accepts if the counter is
zero. 


\begin{figure}[!ht]
		\centering
		\drawLn
	\caption{A $\twoDPA$ recognising $L_n = \{
a^k\# u\mid u\in\{ b, c \}^*\land k = |\{ i\mid 1\leq i\leq
|u|-n\wedge u[i]\neq u[i+n]\}|\}$}
	\label{fig:mismatch}
\end{figure}


Our second example shows how to encode multiplication.
The language $\{ a^n \# a^m\# a^{n \times m}\mid n,m \in \N \}$ is indeed definable by the $\twoNPA$ of Figure~\ref{fig:times} which has dimension $2$.
When reading a word of the form $a^n\#a^m\# a^\ell$, every accepting run makes $p$ passes over $a^n$ where $p$ is chosen non-deterministically by the choice made on state $q_1$ on reading $\#$. 
Along those $k$ passes, the automaton increments the first dimension whenever $a$ is read in a right-to-left pass.
It also counts the number of passes in the second dimension.
Thus, when entering state $q_2$, the sum of the vectors so far is $(np, p)$.
Then, on $a^m$, it decrements the second dimension and on $a^{\ell}$, it decrements the first dimension, and eventually checks that both the counters are equal to zero, which implies that $p = m$ and $\ell = n  p = n  m$.
Note that this automaton is not bounded-visit as its number of visits to any position of $a^n$ is arbitrary.

\begin{figure}[!ht]
	\centering
	\drawTimes
	\caption{A $\twoNPA$ recognising $\{ a^n \# a^m\# a^{n \times m}\mid n,m \in \N\}$} \label{fig:times}
\end{figure}



%
%


	\section{Relating two-way and one-way Parikh automata} \label{sec:two-way_to_one_way}

In this section, we provide an algorithm which converts a bounded-visit $\twoNPA$ into a $\NPA$ defining the same language, through a \emph{crossing section} construction.
This technique is folkloric in the literature (see Section~2.6 of \cite{DBLP:books/aw/HopcroftU79}) and has been introduced to convert a $\twoNFA$ into an equivalent $\NFA$.
Intuitively, the one-way automaton is constructed such that on each position $i$ of the input word, it guesses a tuple of transitions (called crossing section), triggered by the original two-way automaton at the same position $i$ and additionally checks a local validity between consecutive tuples (called \emph{matching} property).
A one-way automaton takes crossing sections as set of states.
Furthermore, the matching property is defined to ensure that the sequence of crossing sections which successively satisfy it, correspond to the sequence of crossing sections of an accepting two-way run.
Thanks to the commutativity of $+$, the order in which weights are
combined by the two-way automaton does not matter and therefore,
transitions of the one-way automaton are labelled by summing the
weights of transitions of the crossing section. Formally, we define a crossing section as follows:

\begin{definition}[crossing section]
	Let $k \in \N_{\neq 0}$.
	Consider a $k$-visit $\twoNPA$ $A$ over $\Sigma$ and $a \in \Sigma \cup \{ \lword, \rword \}$.
	An $a$-crossing section is a sequence $c = (p_1, a, q_1) \dots (p_\ell, a, q_\ell) \in \Delta^+$ such that $1 \leq \ell \leq k$, $p_1, q_\ell \in Q^\onright$ and for all $m \in \{ \onleft, \onright \}$, $p_i \in Q^m \implies p_{i+1} \notin Q^m$.
	We define the value of $c$ as $V(c) = \sum_{i=1}^\ell \lambda(p_i, a, q_i)$, and its length $|c| = \ell$.
	From the sequence $s = p_1q_2p_3\dots q_{\ell-1}p_{\ell}$, the $\onleft$-anchorage of $c$ is defined by $p_1f(q_2,p_3)\dots f(q_{\ell-1},p_{\ell})$ where $f(q_i, p_{i+1}) = \varepsilon$ if $q_i = p_{i+1}$ and $q_i \in Q^\onright$, otherwise $f(q_i, p_{i+1}) = q_ip_{i+1}$.
	The $\onright$-anchorage of $c$ is defined dually\footnote{From $s = q_1p_2\dots q_{\ell-2}p_{\ell-1}q_\ell$, we define $f(q_1p_2)\dots f(q_{\ell-2}p_{\ell-1})q_\ell$ where $f(q_i, p_{i+1}) = \varepsilon$ if $q_i = p_{i+1}$ and $q_i \in Q^\onleft$ otherwise $f(q_i, p_{i+1})$ is the identity}.
	Furthermore, $c$ is said to be \emph{initial} if its $\onleft$-anchorage is $p_1 \in Q_I$.
	Dually, $c$ is said to be \emph{accepting} if its $\onright$-anchorage is $q_\ell \in Q_F$.
\end{definition}

Given a run $\rho$ of a $\twoPA$ over $u$ and a position $1 \leq i
\leq |u|$, the \emph{crossing section of $\rho$ at position $i$} is
defined as the sequence of all transitions triggered by $\rho$ when
reading the $i$th letter, taken in the order of appearance in $\rho$. 
We also define the \emph{crossing section sequence} $\C(r)$ as the sequence of crossing sections of $\rho$ from position $1$ to $|u|$.
Note that the first crossing section is initial and the last crossing section of $\rho$ is accepting if $\rho$ is accepting.

\begin{figure}[!ht]
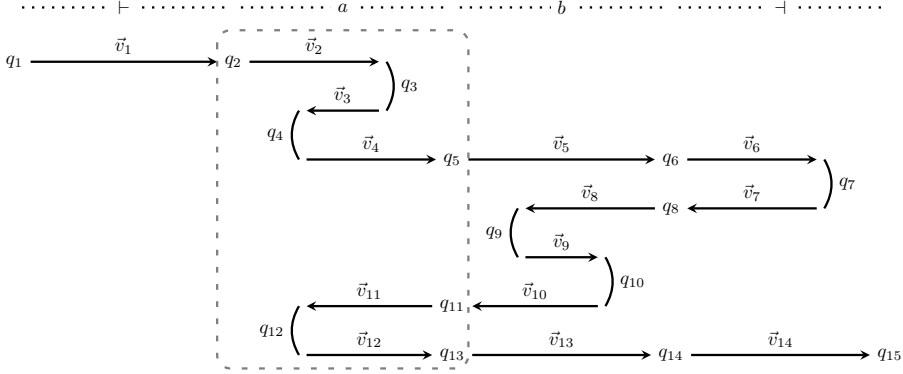

	\centering
	\drawSectionBis
	\caption{A $a$-crossing section of a run} \label{fig:match}
\end{figure}

\begin{example}
	Figure~\ref{fig:match}, shows a run over the word $\lword a b \rword$.
	Consider the $a$-crossing section $c = (q_2, a, q_3)(q_3, a, q_4)(q_4, a, q_5)(q_{11}, a, q_{12})(q_{12}, a, q_{13})$.
	We have that $\onleft$-anchorage of $c$ is $q_2f(q_4,q_4)f(q_{12},q_{12}) = q_2$, $\onright$-anchorage of $c$ is $f(q_3,q_3)f(q_{5},q_{11})q_{13} = q_5q_{11}q_{13}$ and
	$V(c) = \vec{v}_2 + \vec{v}_3 + \vec{v}_4 + \vec{v}_{11} + \vec{v}_{12}$.
	Note that, the states of the crossing section do not appear in the anchorage when the run changes its reading direction.
\end{example}


\begin{definition}[matching relation]
	Consider two crossing sections $c_1, c_2$ from the same automaton.
	The matching relation $M$ is defined such that $(c_1, c_2) \in M$ if the $\onright$-anchorage of $c_1$ equals the $\onleft$-anchorage of $c_2$.
\end{definition}


In general, an arbitrary sequence of crossing sections may not correspond to a run of a two-way automaton, that is 
a crossing section sequence $s = c_1, \dots, c_\ell$ such that $\C(r) \neq s$ for all run $\rho$.
Lemma~\ref{lem:merging} shows that the matching property ensures the existence of such a run $\rho$ in the two-way automaton.

\begin{lemma} \label{lem:merging}
	Consider $s = c_1, \dots, c_n$ where $c_i$ is an $a_i$-crossing section such that $c_1$ is initial, $c_n$ is accepting, and $(c_i, c_{i+1}) \in M$ for all $i \in \{1, \dots, n-1\}$.
	Then there exists an accepting two-way run $\rho$ over $a_1 \dots a_n$ such that $\C(\rho) = s$.
	Moreover, $V(r) = \sum_{i=1}^n V(c_i)$.
\end{lemma}


\begin{theorem} \label{thm:two-way_to_one-way}
	Let $k \in \N_{\neq 0}$.
	Given a $k$-visit $\twoNPA$ $P$, one can effectively construct a language equivalent $\NPA$ $R$ that is at most exponentially bigger.
	Furthermore, if $P$ is deterministic then $R$ is unambiguous.
\end{theorem}
\begin{proof}
	Let $P = (A, \lambda, \psi)$ with $A = (Q, Q_I, Q_H, Q_F, \Delta)$ be a $k$-visit $\twoNPA$ of dimension $d$ with $n = |Q|$ states.
	In this proof we show how to construct $R = (B, \omega, \psi)$ where $B = (V, V_I, V_H, V_F, \Gamma)$ is a $\NPA$ of dimension $d$ having $\O(n^{2k})$ states  such that $|\range(\omega)| \leq |\range(\lambda)|^{k+1}$.
	Note that the formula $\psi$ is the same in both $P$ and $R$.
	
	To do so, we first consider a symbol $\top$ and extend the relation $M$ such that $(c, \top) \in M$ holds for all accepting crossing section $c$.
	Then, we define $R$ as follows:
	\begin{itemize}
	\item
		$V$ is the set of crossing sections of length at most $k$
		
	\item
		$V_I$ is the set of initial crossing sections and $V_H = V_F = \{ \top \}$ 
		
	\item
		$\Gamma = \{ (c_1, a, c_2) \in V \times \Sigma \cup \{\lword, \rword\} \times V \mid (c_1, c_2) \in M \land c_1 \text{ is an $a$-crossing section} \}$
		
	\item
		$\omega \colon (c_1, a, c_2) \mapsto V(c_1)$
	\end{itemize}
Similar to the case of $\twoNFA$, a word $u$ is accepted by $B$ if there exists an accepting run of $B$ on $\lword u \rword$, and the language $L(B)$ of $B$ is defined as the set of words it accepts.
        The inclusion $L(R) \subseteq L(P)$ is a direct consequence of
        Lemma~\ref{lem:merging}, while the other direction is based on
        the following observation: any accepting two-way run $\rho$ has a
        sequence of crossing sections $\C(r)$, consecutively satisfying the
        matching relation.
        Note that, the choice of $c_2$ in a transition $(c_1, a, c_2)$ is non-deterministic in general; but when $P$ is deterministic at most one such choice of $c_2$ will corresponds to a two-way run ensuring unambiguity.
        Details can be found in Appendix. 
\end{proof}


The previous crossing section construction permits to construct a one-way automaton from a bounded-visit two-way one.
This construction is exponential in the number of states and in the number of distinct weight vectors.
Nevertheless, a close inspection of the proof of Theorem~\ref{thm:two-way_to_one-way}, reveals that the exponential explosion in the number of distinct weight vectors can be avoided, while preserving the non-emptiness (but not the language).

\begin{lemma} \label{lem:clever_two-way_to_one-way}
	Let $P$ be a $k$-visit $\twoNPA$.
	We can effectively construct a $\NPA$ $R$ with $\O(n^{2k})$ states and such that $L(R) = \varnothing$ iff $L(P) = \varnothing$.
	Furthermore, $R$ has the same set of weight vectors and the same acceptance constraint as $P$.
\end{lemma}
\begin{proof}
	The construction is the same as in
        Theorem~\ref{thm:two-way_to_one-way} but each transition of
        the one-way automaton $t = (c_1, a, c_2)$ is split into the
        following $|c_1|$
        consecutive transitions, using a fresh symbol
        $\#\notin\Sigma$: 
        $c_1\xrightarrow{a} (t,1)\xrightarrow{\#}
        (t,2)\xrightarrow{\#}\dots (t,|c_1|-2)\xrightarrow{\#} (t,
        |c_1|-1)\xrightarrow{\#} c_2$. The vectors of those transitions are defined as follows. If
        $c_1[i]$ denotes the $i$th transition of $c_1$, then 
        the vector of the first $R$-transition is the vector of the
        $P$-transition $c_1[1]$, and
        the vector of any $R$-transition from state $(t,i)$ is the vector
        of the $P$-transition $c_1[i+1]$. 
	The two languages are then equal modulo erasing $\#$ symbols.
\end{proof}


\begin{theorem} \label{thm:2DPA_UPA}
	Unambiguous Parikh automata have the same expressiveness as
        two-way deterministic (even reversible\footnote{deterministic
          and co-deterministic}) Parikh automata i.e.\ $\UPA = \twoDPA$.
	Furthermore, the transformation from one formalism to the other can be done in \ExpTime.
\end{theorem}
\begin{proof}
	We only show here $\UPA \subseteq \twoDPA$.
	The opposite direction is given by Theorem~\ref{thm:two-way_to_one-way}.
	Let $P = (A, \lambda, \psi)$ be a $\UPA$ of dimension $d$ over $\Sigma$.
	Consider the alphabet $\Lambda \subseteq \Z^d$ as the set of vectors occurring on the transitions of $P$.
	We can see the automaton $A$ with the morphism $\lambda$ as an unambiguous finite transducer $T$ defining a function from $\Sigma^*$ to $\Lambda^*$.
	It is known that any unambiguous letter-to-letter one-way transducer can be transformed into an equivalent letter-to-letter deterministic two-way transducer.
	This result is explicitly stated in Theorem~1 of~\cite{Chytil}
	which is based on a general technique introduced by Aho, Hopcroft and Ullman~\cite{Aho1969}\footnote{Based on AHU's technique, a similar result was shown in~\cite{DBLP:journals/ijfcs/CarninoL15} for weighted automata, namely that unambiguous weighted automata over a semiring can be equivalently converted into deterministic two-way weighted automata}.
	Recently, another technique has been introduced which improves AHU's technique by one exponential~\cite{DBLP:conf/icalp/DartoisFJL17}, and allows to show that any unambiguous finite transducer is equivalent to a reversible two-way transducer exponentially bigger, yielding our result.
\end{proof}


	\section{Emptiness Problem} \label{sec:decision}

The emptiness problem asks, given a $\twoNPA$, whether the language it accepts is empty.
We have seen in Example~\ref{fig:times} how to encode the multiplication of two natural numbers encoded in unary.
We can generalise this to the encoding of solutions of Diophantine equations as languages of $\twoNPA$, yielding undecidability:

\begin{theorem}\label{thm:undectwoNPA}
	The emptiness problem for $\twoNPA$ is undecidable.
\end{theorem}


The proof of this theorem relies on the fact that an input position can be visited an arbitrary number of times, due to non-determinism.
If instead we forbid this, we recover decidability.
To prove it, we proceed in two steps: first, we rely on the result of the previous section showing that any bounded-visit $\twoNPA$ can be effectively transformed into some (one-way) $\NPA$.
This yields decidability of the emptiness problem as this problem is known to be decidable for $\NPA$.
To get a tight complexity in \PSpace, we analyse this transformation (which is exponential), to get exponential bounds on the size of shortest non-emptiness witnesses. 
A key lemma is the following, whose proof gathers ideas and arguments that already appeared in~\cite{phdLin,DBLP:conf/lics/FigueiraL15}.
Since the statement was not explicit in those papers, and its proofs relies on arguments that appear at different places, we prove it in Appendix.

\begin{lemma}\label{lem:one-way_to_formula}
	Let $P$ be a one-way Parikh automaton with $n$ states and $\gamma$ distinct weight vectors.
	Then, we can construct an existential Presburger formula
        $\varphi(x) = \bigvee_{i=1}^m \varphi_i(x)$ such that
        for all $\ell\in \N$,  $\varphi(\ell)$ holds iff there exists $w \in L(P) \cap \Sigma^{|\ell|}$.
	Furthermore, $log_2(m)$ and each $\varphi_i$ are 
$O(\poly(|P|, \log n))$,
	and can be constructed in time $2^{\O(\gamma^2\log(\gamma n))}$.
\end{lemma}


Thanks to the lemma above, we are able to show that the non-emptiness
problem for bounded-visit $\twoNPA$ is \PSpaceC, just as the
non-emptiness problem for two-way automata. In some sense, adding
semi-linear constraints to two-way automata is for free as long as it is
bounded-visit. 

\begin{theorem} \label{thm:PA_emptiness}
	The non-emptiness problem for bounded-visit $\twoNPA$ is \PSpaceC.
	It is \NPTimeC for $k$-visit $\twoNPA$ when $k$ is fixed.
\end{theorem}
\begin{proof}
	Consider a $k$-visit $\twoNPA$ $P = (A, \lambda, \psi)$ of dimension $d$.
	We start with the \PSpace membership of the non-emptiness problem for $\twoDFA$.
	Intuitively, we first want to apply Lemma~\ref{lem:clever_two-way_to_one-way} in order to deal with a one-way automaton, and apply then Lemma~\ref{lem:one-way_to_formula} to reduce the non-emptiness problem of the one-way Parikh automaton to the satisfiability of an existential Presburger formula.
	Nevertheless, we cannot explicitly transform $P$ into a one-way automaton while keeping polynomial space.
	So, in the sequel, $(i)$ we highlight an upper bound on the	smallest witness of non-emptiness
	and based on it, $(ii)$ we provide an \NPSpace algorithm which decides if there exists such a witness.
	
	$(i)$
	By Lemma~\ref{lem:one-way_to_formula} applied on the $\NPA$ obtained from Lemma~\ref{lem:clever_two-way_to_one-way}, there exists an existential Presburger formula $\varphi(\ell) = \bigvee_{i=1}^m \varphi_i(\ell)$ where each $\varphi_i$ is polynomial in $|P|$.
	This formula is satisfiable iff there exists $w \in \Sigma^{|\ell|}$ such that $w \in L(P)$.
	By Theorem~6~(A) of~\cite{scarpellini}, there exists $N$ exponential in $|\varphi_i|$ such that $\varphi_i$ is satisfiable iff $\varphi_i(\ell)$ holds for some $0 \leq \ell \leq N$.
	Hence, there exists $N$ exponential in $|P|$ such that $\min\{ |u| \mid u \in L(P) \} \leq N$.
	
	$(ii)$ The algorithm guesses a witness $u$ of length at most $N$ on-the-fly and a run on it.
	It controls its length by using a binary counter: as $N$ is exponential in $|P|$, the memory needed for that counter is polynomial in $|P|$.
	The transitions of the one-way automaton obtained from Lemma~\ref{lem:clever_two-way_to_one-way} can also be computed on-demand in polynomial space.
	Eventually, it suffices to check the last state is accepting
        and the sum $\vec{v} = (v_1,\dots,v_d)$ of the vectors
        computed on-the-fly along the run, satisfies the
        Presburger formula $\psi(x_1,\dots,x_d)$. To do so, our
        algorithm constructs a closed formula $\psi^{\vec{v}}$ in polynomial
        time such that $\psi^{\vec{v}}$ is true iff $\vec{v}\models
        \psi$. To do so, it hardcodes the values of $\vec{v}$
        in $\psi$ by substituting each $x_i$ by a term $t_{v_i}$ of size
        $(\text{log}_2(v_i))^2$ encoding $v_i$, by using the function
        symbol $\times_2$. E.g. $t_{13} =
        \times_2(\times_2(\times_2(1))) + \times_2(\times_2(1)) + 1$. 
        Let us argue that $\psi^{\vec{v}}$ has polynomial size. Let $\mu$ be the maximal absolute entry
        of vectors of $P$,
        then $v_i\leq \mu N$, and since $N$ is
        exponential in $|P|$, $t_{v_i}$ has polynomial size in $|P|$ and
        $\log_2(\mu)$. Hence $\psi^{\vec{v}}$ has polynomial size, and
        its satisfiability can be checked in \NPTime~\cite{scarpellini}.

        The lower bound is direct as it already holds for the
        emptiness problem of deterministic two-way automata, by a trivial
        encoding of the \PSpaceC intersection problem of $n$ $\DFA$~\cite{Dexter}.

	When $k$ is fixed, then the conversion to a one-way automaton (Lemma~\ref{lem:clever_two-way_to_one-way}) is polynomial.
	Then, the result follows from the \NPTimeC result for the non-emptiness of $\NPA$~\cite{DBLP:conf/lics/FigueiraL15}.
\end{proof}

\begin{remark}
	In~\cite{DBLP:conf/lics/FigueiraL15}, non-emptiness is shown to be
        polynomial time for $\NPA$ when the dimension is fixed, the
        values in the vectors are unary encoded and the semi-linear
        constraint is period-base represented. 
	As a consequence, for all fixed $d,k$, the non-emptiness problem
        for $k$-visit $\twoNPA$ with vectors in $\{0,1\}^d$ and a
        period-base represented semi-linear constraint can be
        solved in \PTime.
\end{remark}


	\section{Closure properties and comparison problems}\label{sec:closure}

Since the class of $\twoDPA$ is equivalent to the class of $\UPA$ that is known to be closed under Boolean operations \cite{DBLP:conf/dlt/CadilhacFM12, KlaRue03}, we get the closure properties of $\twoDPA$ for free, although with non-optimal complexity.
We show here that they can be realised in linear-time for intersection and union, and with linear state-complexity for the complement.

\begin{theorem}[Boolean closure]\label{thm:twoDPAclosure}
  Let $P, P_1, P_2$ be $\twoDPA$ such that $P = (A, \lambda, \psi)$.
  One can construct a $\twoDPA$ $\overline{P} = (A',\lambda',\psi')$ such that $L(\overline{P}) = \overline{L(P)}$ and the size of $A'$ is linear in the size of $A$.
  One can construct in linear-time a $\twoDPA$ $P_\cup$ (resp. $P_{\cap}$) such that $L(P_\cup) = L(P_1)\cup L(P_2)$ (resp.\ $L(P_\cap) = L(P_1)\cap L(P_2)$).
\end{theorem}

\begin{proof}
	Let us start by intersection, assuming $P_i = (A_i,\lambda_i,\psi_i)$ has dimension $d_i$.
	The automaton $P_\cap$ is constructed with dimension $d_1+d_2$.
	Then $P_\cap$ first simulates $P_1$ on the first $d_1$ dimensions (with weight vectors belonging to $\Z^{d_1}\times \{0\}^{d_2}$), and then, if $P_1$ eventually reaches an halting state, it stops if it is non-accepting and reject, otherwise it simulates $P_2$ on the last $d_2$ dimensions with vectors in $\{0\}^{d_1}\times \Z^{d_2}$, and accepts the word if the word is accepted by $P_2$ as well.
	The Presburger acceptance condition is defined as $\psi(\vec{x}_1, \vec{x}_2)=\psi_1(\vec{x}_1) \land \psi_2(\vec{x}_2)$.
	Note that if $P_1$ never reaches an halting state, then $P_\cap$ won't either, so the word is rejected by both automata.
	It is also a reason why this construction cannot be used to show closure under union: even if $P_1$ never reaches an halting state, it could well be the case that $P_2$ accepts the word, but the simulation of $P_2$ in that case will never be done.
	However, assuming that $P_1$ halts on any input, closure under union works with a similar construction.
	Additionally, we need to keep in some new counter $c$ the information whether $P_1$ has reached an accepting state: First $P_\cup$ simulates $P_1$, if $P_1$ halts in some accepting state, then $c$ is incremented and $P_\cup$ halts, otherwise $P_\cup$ proceeds with the simulation of $P_2$.
	The formula is then $\psi(\vec{x}_1, \vec{x}_2, c) = (c=1 \land \psi_1(\vec{x}_1)) \lor \psi_2(\vec{x}_2)$.
	
	So, we have closure under union in linear-time as long as $P_1$ halts on every input.
	This can be used to show closure under complement, using the following observation: $\overline{L(P)} = \overline{L(A)}\cup L(A,\lambda,\lnot \psi)$ and moreover, it is known that $\twoDFA$ can be complemented in linear-time into a $\twoDFA$ which always halts~\cite{GMP07}.
	The formula $\lnot \psi$ is universal since $\psi$ is
        existential. Then, $\lnot \psi$ could be converted into an
        equivalent existential formula
        using quantifier elimination~\cite{Cooper}.
	
	For the closure under union, we use the equality $L(P_1)\cup L(P_2) = \overline{\overline{L(P_1)}\cap \overline{L(P_2)}}$.
	It can be done in linear-time because the formulas for $\overline{P_1}$ and $\overline{P_2}$ are universal, and so is the formula for the $\twoDPA$ accepting $\overline{L(P_1)}\cap \overline{L(P_2)}$.
	By applying again the complement construction, we get an
        existential formula (without using quantifier eliminations). 
\end{proof}

Thanks to Theorem~\ref{thm:twoDPAclosure} and decidability of non-emptiness for $\twoDPA$, we easily get the decidability of the universality problem (deciding whether $L(P) = \Sigma^*$), the inclusion problem (deciding whether $L(P_1)\subseteq L(P_2)$), and the equivalence problem (deciding whether $L(P_1) = L(P_2)$) for $\twoDPA$.
The following theorem establishes tight complexity bounds.
It is a consequence of a more general result (Theorem~\ref{thm:comparisongen}) that we establish for Parikh automata with \emph{arbitrary} Presburger formulas in Section~\ref{sec:generalised}.

\begin{theorem}[Comparison Problems]\label{thm:comparisons}
	The universality, inclusion and equivalence problems are \coNExpTimeC for $\twoDPA$.
\end{theorem}

Finally, we study the membership problem which asks given a Parikh
automaton $P$ and a word $w \in \Sigma^*$, whether $w \in
L(P)$. Hardness was known already for
$\NPA$~\cite{DBLP:conf/lics/FigueiraL15}.

\begin{theorem} \label{thm:membership}
	The membership problem for $\twoNPA$ is \NPTimeC.
\end{theorem}


	\section{Parikh automata with arbitrary Presburger acceptance condition} \label{sec:generalised}

In this section, we consider Parikh automata where the acceptance constraint is given as an arbitrary Presburger formula, that is, not restricted to existential Presburger formula, and we study the complexity of their decision problems.
For all $i > 0$, a two-way $\SIGMA_i$-Parikh automaton ($\SIGMA_i$-$\twoNPA$ for short) is a tuple $P = (A, \lambda, \Psi)$ where $A,\lambda$ are defined just as for $\twoNPA$ and $\Psi\in\SIGMA_i$.
In particular, a $\SIGMA_{1}$-$\twoNPA$ is exactly a $\twoNPA$.
Similarly, we also define $\SIGMA_{i}$-$\DPA$, $\SIGMA_{i}$-$\twoDPA$, $\SIGMA_{i}$-$\NPA$ respectively, and their $\PI_i$ counterpart (when the formula is in $\PI_i$).

The complexity of Presburger arithmetic has been connected to the weak \ExpTime hierarchy~\cite{Hemachandra89, Haase14} which resides between \NExpTime and \ExpSpace is defined as $\bigcup_{i \geq 0}{\SIGMA^{\textsc{Exp}}_i}$ where:
$$
	\begin{array}{lll}
		\SIGMA^{\textsc{P}}_0 \new \PI^{\textsc{P}}_0 \new \PTime
	&
		\SIGMA^{\textsc{P}}_{i+1} \new \NPTime^{\SIGMA^\textsc{P}_i}
	&
		\PI^{\textsc{P}}_{i+1} \new \coNPTime^{\SIGMA^\textsc{P}_i}
	\\
		\SIGMA^{\textsc{Exp}}_0 \new \PI^{\textsc{Exp}}_0 \new \ExpTime
	&
		\SIGMA^{\textsc{Exp}}_{i+1} \new \NExpTime^{\SIGMA^\textsc{P}_i}
	&
		\PI^{\textsc{Exp}}_{i+1} \new \coNExpTime^{\SIGMA^\textsc{P}_i}
	\end{array}
$$



Since Lemma~\ref{lem:one-way_to_formula} uses the acceptance constraint as a black box, we can generalise it as follows.

\begin{lemma}\label{lem:generalised_one-way_to_formula}
	For any fixed $i \in \N_{\neq 0}$, given a $\SIGMA_{i}$-$\NPA$ $P$ with $n$ states and $\gamma$ distinct weight vectors,
	we can construct a $\SIGMA_{i}$-formula $\Phi$ such that for all $\ell\in \N$ we have that $\Phi(\ell) = \bigvee_{j=1}^m \Phi_j(\ell)$ holds iff there exists $w \in L(P) \cap \Sigma^{|\ell|}$.
	Furthermore, $\log_2(m)$ and size of each $\Phi_j$ are 
	$\poly(|P|, \log n)$,
	and can be constructed in time $2^{\O(\gamma^2\log(\gamma n))}$.
\end{lemma}

Using Lemma \ref{lem:generalised_one-way_to_formula},
we can extend Theorem~\ref{thm:PA_emptiness} to bounded-visit $\SIGMA_{i+1}$-$\twoNPA$.
Note that the case of $\SIGMA_{1}$-$\twoNPA$ is not covered by the following statement.

\begin{theorem} \label{thm:generalised_emptiness}
	For any fixed $i \in \N_{\neq 0}$, the non-emptiness problem for bounded-visit $\SIGMA_{i+1}$-$\twoNPA$ is \textsc{$\SIGMA_i^\textsc{Exp}$-C}.
\end{theorem}
\begin{proof}
	For the upper-bound, we show that this problem can be solved by an alternating Turing machine in exponential time, which alternates at most $i$ times between sequences of non-deterministic and universal transitions, starting with non-deterministic transitions.
	By~\cite{Haase14}, the satisfiability of $\SIGMA_{i+1}$-formulas is complete for \textsc{$\SIGMA_i^\textsc{Exp}$-C}.
	Hence there is an $i$-alternating machine $\mathcal{M}$ running in exponential time which checks the satisfiability of such formulas. 
	Now, similar to the case of $\SIGMA_{1}$ in Theorem~\ref{thm:PA_emptiness}, from a bounded-visit $\SIGMA_{i+1}$-$\twoNPA$ $P$ one can construct a $\SIGMA_{i+1}$-formula which is true iff the automaton has a non-empty language.
	We can do so by applying Lemma~\ref{lem:generalised_one-way_to_formula} on the $\NPA$ obtained\footnote{Lemma~\ref{lem:clever_two-way_to_one-way} can be trivially adapted to $\SIGMA_{i}$-formula as acceptance condition} from Lemma~\ref{lem:clever_two-way_to_one-way}.
	Hence, non-emptiness of a bounded-visit $\SIGMA_{i+1}$-$\twoNPA$ reduces to satisfiability of a $\SIGMA_{i+1}$-formula $\Phi(\ell) = \bigvee_{j=1}^m \Phi_j(\ell)$ such that $\log_2(m)$ and the size of each $\Phi_j$ are polynomial in $|P|$ and can be constructed in time $2^{\O(\gamma^2\log(\gamma n))}$.
	However we cannot construct explicitly $\Phi$, since its size is exponential in $|P|$.
	Instead we construct an $i$-	alternating machine $\mathcal{M}'$ that first guesses a disjunct $\Phi_s$ and constructs it in exponential time, and then simulates the machine $\mathcal{M}$ on $\Phi_s$.
	Recall the $\mathcal{M}$ starts with non-deterministic transitions.
	Thus the machine $\mathcal{M}'$ runs in exponential time, and also performs only $i$ alternations, which provides \textsc{$\SIGMA_i^\textsc{Exp}$} upper bound.
	
	Hardness comes from checking if a $\SIGMA_{i+1}$-sentence holds true, which is \textsc{$\SIGMA_{i}^\textsc{Exp}$-C} by~\cite{Haase14}.
	From a $\SIGMA_{i+1}$-sentence $\Psi$ it suffices to construct a Parikh automaton $P = (A, \lambda, \Psi)$ of dimension $0$ such that $L(A) \neq \varnothing$, therefore $L(P)\neq \varnothing$ iff $L(P) = L(A)$ iff $\Psi$ holds.
\end{proof}

\begin{theorem}[Boolean closure] \label{thm:gentwoDPAclosure}
	Let $P,P_1,P_2$ be $\SIGMA_i$-$\twoDPA$.
	One can construct in linear time a $\PI_i$-$\twoDPA$ $\overline{P}$ and two $\SIGMA_i$-$\twoDPA$ $P_\cup,P_{\cap}$ such that $L(\overline{P}) = \overline{L(P)}$, $L(P_\cup) = L(P_1)\cup L(P_2)$ and $L(P_\cap) = L(P_1)\cap L(P_2)$.
\end{theorem}
\begin{proof}
	The constructions are the same as in the proof of the case $i=1$ of Theorem~\ref{thm:twoDPAclosure}, using closure under disjunction and conjunction of $\SIGMA_i$ and the fact that negating a $\SIGMA_i$-formula yields a $\PI_i$-formula.
\end{proof}

\begin{theorem}[Comparison Problems] \label{thm:comparisongen}
	For all fixed $i \in \N_{\neq 0}$, the universality, inclusion and equivalence problems for $\SIGMA_i$-$\twoDPA$ are \textsc{$\PI_i^{\textsc{Exp}}$-C}.
\end{theorem}

\begin{proof}
	We first prove the upper bound for the most general problem which is inclusion.
	Let $P_i = (A_i,\lambda_i,\psi_i)$ be a $\SIGMA_{i}$-$\twoDPA$.
	Note that $L(P_1) \subseteq L(P_2)$ iff $L(P_1) \cap \overline{L(P_2)} = \varnothing$.
	So, using Theorem~\ref{thm:gentwoDPAclosure} we first construct in linear-time a $\PI_i$-$\twoDPA$ $\overline{P_2} = (A_2', \lambda_2', \Psi'_2)$ such that $L(\overline{P_2}) = \overline{L(P_2)}$ and then $P_{\cap}=(A, \lambda, \Psi)$ such that $L(P_\cap) = L(P_1) \cap L(\overline{P_2})$.
	From the construction in Theorem~\ref{thm:twoDPAclosure} generalised to $\SIGMA_i$-$\twoDPA$, recall that the formula $\Psi$ is defined as $\Psi(\vec{x}_1, \vec{x}_2)=\Psi_1(\vec{x}_1) \land \Psi_2'(\vec{x}_2)$.
	Let
		$\Psi_1(\vec{x}_1) = \exists \vec{y}_1 \forall \vec{y}_2 \dots \Omega \vec{y}_i \left[
			\varphi_1(\vec{x}_1, \vec{y}_1, \dots, \vec{y}_i)
		\right]$, and
		$\Psi_2'(\vec{x}_2) = \forall \vec{z_1} \exists \vec{z_2}\dots \rotatebox[origin=c]{180}{$\Omega$} \vec{z_i} \left[	
			\varphi_2(\vec{x}_2, \vec{z}_1, \dots, \vec{z}_i)
		\right]$
	where $\Omega, \rotatebox[origin=c]{180}{$\Omega$} \in \{ \exists, \forall \}$ such that  $\Omega \neq \rotatebox[origin=c]{180}{$\Omega$}$.
	Hence $\Psi$ is equivalent to the following $\SIGMA_{i+1}$-formula.
	$$
		\exists \vec{y}_1 \forall \vec{z}_1 \forall \vec{y}_2 \exists \vec{z}_2 \exists \vec{y}_3 \dots \Omega \vec{z}_{i-1} \vec{y}_i \rotatebox[origin=c]{180}{$\Omega$} \vec{z}_i
		\Big[
			\varphi_1(\vec{x}_1, \vec{y}_1, \dots, \vec{y}_i) \land \varphi_2(\vec{x}_2, \vec{z}_1, \dots, \vec{z}_i)
		\Big]
	$$
	Finally, emptiness of $P_{\cap}$ can be decided in $\PI_i^{\textsc{Exp}}$ by Theorem~\ref{thm:generalised_emptiness}.

	For the lower bound, we show that the universality problem of $\SIGMA_{i}$-$\DPA$ is $\PI_i^{\textsc{Exp}}$-hard.
	This holds even for a fixed number of states and vector values in $\{-1,0,1\}$, showing that the complexity comes from the formula part.
	From a $\SIGMA_i$-formula $\Psi$ with $d$ free variables, we construct a Parikh automaton $P = (A,\lambda,\Psi)$ of dimension $d$ over alphabet $\Sigma = \{a^+_i, a^-_i\}_{1\leq i\leq d}$.
	Any word $w$ over $\Sigma$ defines a valuation $\mu_w(x_i) = |w|_{a^+_i} - |w|_{a^-_i}$ for all $1 \leq i \leq d$. Conversely, any valuation $\mu$ can be encoded as a word over $\Sigma$.
	Hence, $\Psi$ holds for all values iff for all $w\in\Sigma^*$, we have $\mu_w \models \Psi$.
	We construct a deterministic one-way automaton $A$ such that $L(A) = \Sigma^*$ and for all $w\in\Sigma^*$, the value of the run $r$ over $w$ is $\mu_w$.
	The automaton $A$ has one accepting and initial state $q$ over which it loops and, when reading $a_i^+$ (resp.\ $a_i^-$) it increases dimension $i$ by $1$ (resp.\ by $-1$).
\end{proof}

\begin{remark}
	Since a $\twoDPA$ is a $\SIGMA_1$-$\twoDPA$, and the class \coNExpTime is the same as $\PI_1^{\textsc{Exp}}$, we have that Theorem~\ref{thm:comparisongen} for $i=1$ is exactly the same as Theorem~\ref{thm:comparisons}.
\end{remark}

	\section{Conclusion} \label{sec:conclusion}
		In this paper, we have provided tight complexity bounds for the
emptiness, inclusion, universality and equivalence problems for
various classes of two-way Parikh automata. We have shown that when
the semi-linear constraint is given as a $\SIGMA_i$-formula, for
$i>1$, the complexity of those problems is dominated by the complexity
of checking satisfiability or validity of $\SIGMA_i$-formulas. We have
shown that $\twoDPA$ (resp.\ bounded-visit $\twoNPA$) have the same expressive
power as unambiguous (one-way) $\PA$ (resp.\ non-deterministic $\PA$). In
terms of succinctness, it is already known that $\twoDFA$ are
exponentially more succinct than $\NFA$, witnessed for instance by the
family $D_n = \{ uu\mid u\in\{0,1\}^*\land |u|=n\}$. However $D_n$ 
is accepted by a $\NPA$ with polynomially many states in $n$, using
$2n$ vector dimensions to store the letters of its input, then checked
for equality using the acceptance constraint. We conjecture that $\twoDPA$ are exponentially
more succinct than $\NPA$, witnessed by the language $L_n$ of
Section~\ref{sec:model}. We leave as future work the
introduction of techniques allowing to prove such results (pumping
lemmas), as the dimension and acceptance constraint size has to be
taken into account as well, as shown with $D_n$. 

Finally, we plan to extend the
pattern logic of~\cite{patternLogic}, which intensively uses (one-way)
Parikh automata for its model-checking algorithm, 
to reason about structural properties of
two-way machines, and use two-way Parikh automata emptiness checking
algorithm for model-checking this new logic.


	\newpage
	\appendix
	\section{Section~\ref{sec:model}: Two-way Parikh automata}

\begin{proof}[Proof of Theorem~\ref{thm:two-way_to_one-way} (continued)]
	We prove now that $L(R) \subseteq L(P)$.
	Consider $u \in L(R)$ and let $r$ be an accepting run of $R$ over $u$ with $s = c_1, \dots, c_m$ the sequence of states visited by $r$ to reach $\top$.
	By Lemma~\ref{lem:merging}, there exists an accepting run $\rho$ of $P$ over $u$ such that $\C(\rho) = s$.
	Moreover, $V(\rho) = \sum_{i=1}^m V(c_i) = V(r)$.
	Hence $u \in L(P)$ since $P$ have the same acceptance constraint as $R$.
	
	We prove now that $L(P) \subseteq L(R)$.
	Consider $u = a_1\dots a_m \in L(P)$ and let $\rho$ be an accepting two-way run of $P$ over $u$ with $\C(\rho) = c_1, \dots, c_m$ i.e.\ $c_i$ is the $a_i$-crossing sections of $\rho$.
	Since $\rho$ is accepting then $c_1$ is initial, $c_m$ is accepting and $(c_i, c_{i+1}) \in M$.
	Furthermore, the $k$-visitness of $P$ implies that each $c_i$ have length at most $k$.
	So, there exists an accepting run $r$ of $R$ over $u$ which visit the sequence of states $c_1, \dots, c_m, \top$.
	Moreover, $V(r) = \sum_{i=1}^m c_i = V(\rho)$.
	Hence $u \in L(R)$ since $R$ have the same acceptance constraint as $P$.
	
	We prove now that if $P$ is deterministic then $R$ is unambiguous by contrapositive.
	Let $r_1, r_2$ be two distinct accepting runs of $R$ over some word $u$ with $s_1$ and $s_2$ be the respective sequences of state states visited by $r_1$ and $r_2$ to reach $\top$.
	Since $r_1 \neq r_2$ then $s_1 \neq s_2$.
	By Lemma~\ref{lem:merging}, there exist $\rho_1, \rho_2$ two accepting runs of $P$ over $u$ such that $\C(\rho_1) = s_1$ and $\C(\rho_2) = s_2$.
	Furthermore $\C(\rho_1) \neq \C(\rho_2)$ implies that $\rho_1 \neq \rho_2$.
	Hence $P$ is not deterministic.
\end{proof}

\section{Section~\ref{sec:decision}: Emptiness Problem}

\begin{proof}[Proof of Theorem~\ref{thm:undectwoNPA}]
	We reduce the problem of deciding whether a system of diophantine equations $S$ over a finite set of variables $X=\{x_1,\dots,x_n\}$ has a solution in $\N$, which is known to be undecidable.
	Each equation is of the form $p_1=p_2$ where $p_1,p_2$ are polynomials over $X$, whose coefficient are assumed to be in $\N$.
	A valuation $\nu$ is mapping $\nu:X\rightarrow \N$.
	We denote by $\nu(p)$ the value of polynomial $p$ under valuation $\nu$.
	We first explain how to encode the values $\nu(p)$ for all $\nu$ as a language and show how to define it with a $\twoNPA$.

	Given a polynomial $p$, we let $sub(p)$ all the subpolynomials appearing in $p$, which is inductively defined by $sub(p_1+p_2) = sub(p_1)\cup sub(p_2) \cup \{p_1+p_2\}$, $sub(p_1 \times p_2) = sub(p_1)\cup sub(p_2)\cup \{p_1 \times p_2\}$, $sub(a) = \{a\}$ for $a\in\N$.

	Given a polynomial $p$ over $X = \{x_1,\dots,x_n\}$, we let $\Sigma_p = \{ 0_x\mid x\in X\}\cup \{1_{p'}\mid p'\in sub(p)\}$ be a finite alphabet.
	Note that if $p_2\in sub(p_1)$, then $\Sigma_{p_2}\subseteq \Sigma_{p_1}$.

	Given a word $w\in\Sigma_p^*$, we let $\nu_w$ the valuation $\nu_w(x) = |0_x|_w$.
	We say that $w$ is a \emph{$\nu$-encoding} of $p$ if $\nu = \nu_w$ and $\nu(p) = |1_p|_w$.
	A language $L\subseteq \Sigma_p^*$ is a \emph{good encoding} of $p$ if for all $\nu$ solution of $p$, there exists a $\nu$-encoding of $p$ in $L$ and conversely, any $w\in L$ is a $\nu$-encoding of $p$ for some $\nu$.
	We now show by induction on $p$ that there exists good encoding $L_p$ of $p$ definable by a $\twoNPA$ $A_p$.
	\begin{enumerate}
	\item
		if $p = a\in\N$ is a given constant, then we let $A_p$ be a finite automaton accepting any word $w\in \Sigma_a^*$ such that $|1_a|_w = a$.
		It has $a$ states.
		
	\item
		if $p=p_1+p_2$, then we let $A_{p_1}$ and $A_{p_2}$ be the two $\twoNPA$ constructed inductively on $p_1$ and $p_2$, assumed to be of dimension $d_1$ and $d_2$ respectively.
		Then, $A_p$ is constructed as follows: it works on alphabet $\Sigma_p$ and has dimension $d = d_1+d_2+3$.
		It first simulates $A_{p_1}$  on the first $d_1$ dimensions (with vector updates in $\Z^{d_1}\times \{0\}^{d_2+3}$), ignoring letters in $\Sigma_{p}\setminus \Sigma_{p_1}$ until it reaches an halting state $q$.
		If $q$ is rejecting, $A_p$ rejects, otherwise it goes back to the beginning of the word and simulates $A_{p_2}$ on next $d_2$ dimensions (with updates in $\{0\}^{d_1} \times \Z^{d_2}\times \{0\}^3$), ignoring letters in $\Sigma_{p}\setminus \Sigma_{p_2}$, until it reaches an halting state $q'$.
		If $q'$ rejects, $A_p$ rejects, otherwise it goes back to the beginning, count with a one-way pass the number occurrences of symbols $1_{p_1}$, $1_{p_2}$ and $1_{p_1+p_2}$ respectively in three counters $x_{p_1}$,  $x_{p_2}$ and $x_{p_1+p_2}$ corresponding to the last three dimensions The semi-linear condition is then given by the formula $\varphi(x_1,\dots,x_d) = \varphi_1(x_1,\dots,x_{d_1})\wedge \varphi_2(x_{d_1+1},\dots,x_{d_1+d_2})\wedge x_{p_1}+x_{p_2} = x_{p_1+p_2}$.
		By construction and induction hypothesis, $A_p$ is a good encoding of $p$.
		
		\item if $p = p_1 \times p_2$, then $A_p$ is on alphabet $\Sigma_p$ and has dimension $d_1+d_2+4$.
		Initially it works as $A_{p_1+p_2}$ during the two first phases (simulation of $A_{p_1}$ followed by simulation of $A_{p_2}$).
		After those two simulations, $A_p$ enters phase $2$, during which it makes $k$ passes over the whole input, where $k$ is chosen non-deterministically (by using the non-determinism of $\twoNPA$).
		On each of these passes, it counts the number of occurrences of symbol $1_{p_2}$ in some counter $x_{mult}$ (intended at the end to contain the value of $p_1 \times p_2$).
		At the end of each pass, it increments by one a counter $x_{pass}$.
		It non-deterministically decides to move to phase $3$ during which it also makes a last pass over the whole input to count the number of occurrences of $1_{p_1}$ and $1_{p_1\times p_2}$ in some counters $x_{p_1}$ and $x_{p_1+p_2}$ and accepts.
		The acceptance formula is then: $\varphi(x_1,\dots,x_d) = \varphi_1(x_1,\dots,x_{d_1})\wedge \varphi_2(x_{d_1+1},\dots,x_{d_1+d_2})\wedge x_{pass} = x_{p_1}\wedge x_{mult} = x_{p_1 \times p_2}$.
		If $A_p$ makes $k$ passes during phase $2$ on input $w$, then we know that the value of $x_{mult}$ is equal to $k \times \nu_w(p_2)$.
		The formula also requires that $k = \nu_w(p_1)$ which leads to the result.
\end{enumerate}

	To encode an equation $p_1 = p_2$, we first construct $A_{p_1}$ and $A_{p_2}$, then construct a $\twoNPA$ $A_{p_1=p_2}$ which first simulates $A_{p_1}$ and $A_{p_2}$ on input $w\in(\Sigma_{p_1}\cup \Sigma_{p_2})^*$, and then performs a last pass where it counts the number of occurrences of $1_{p_1}$ in some counter $x_{p_1}$, and similarly for $1_{p_2}$ in some counter $x_{p_2}$.
	The final formula also requires that $x_{p_1} = x_{p_2}$.
	Then $L(A_{p_1=p_2})\neq\varnothing$ iff there exists a solution to $p_1=p_2$.
	It can be easily generalised to a system of equations.
\end{proof}

To prove Lemma~\ref{lem:one-way_to_formula} one needs the following result:

\begin{theorem}[Theorem~7.3.1 of~\cite{phdLin}] \label{thm:phdLin}
	Let $A$ be an NFA with $n$ states over an alphabet $\Lambda$ of size $\gamma$.
	Then, the Parikh image $\P(L(A))$ is equal to the semi-linear set
	$\bigcup_{i=1}^m \{\vec{b_i}  + \sum_{j=1}^\gamma x_{i, j}\vec{p_{i, j}}  \mid x_{i, j} \in \N\}$
	where $m \leq n^{\gamma^2+3\gamma+3}\gamma^{4\gamma+6}$, $||\vec{b_i}|| \leq n^{3\gamma+3}\gamma^{4\gamma+6}$ and $\vec{p_{i, j}} \in \{0,\dots,n\}^\gamma$.
	Furthermore, $b_i$ and $p_{i, j}$ can be computed in time $2^{\O(\gamma^2\log(\gamma n))}$.
\end{theorem}

\begin{proof}[Proof of Lemma~\ref{lem:one-way_to_formula}]
	Let $P = (A, \lambda, \psi)$ be a $\NPA$ of dimension $d$ over $\Sigma$ with the $\NFA$ $A = (Q, \Delta, I, F)$.
	Consider the alphabet $\Lambda \subseteq \Z^d$ the set of vectors occurring on the transitions of $P$ with an arbitrary order defined as $\Lambda = \{ \vec{a}_1, \dots, \vec{a}_\gamma\}$ where $\vec{a}_i \in \range(\lambda) $ and $\gamma = |\range(\lambda)|$.
	We construct the $\NFA$ $A_\lambda$ over $\Lambda$ from $P$ which takes weight vectors as letters instead of $\Sigma$.
	Formally, $A_\lambda = (Q, \Delta_\lambda, I, F)$ such that $(p, \lambda(p, a, q), q) \in \Delta_\lambda$ iff $(p, a, q) \in \Delta$.

	This proof shows the existence of the existential Presburger formula $\varphi(\ell)$ which holds iff $L(P) \neq \varnothing$.
	To do that we consider $\P(L(A_\lambda)) \subseteq \N^\gamma$, the Parikh image of $L(A_\lambda)$, assuming that it can be denoted by the existential Presburger formula $\xi$.
	Indeed, $(\tau_1, \dots, \tau_\gamma) \in \P(L(A_\lambda))$ iff there exists an accepting run $\rho$ of $A_\lambda$ which visits each weight vector $\vec{a}_i \in \Lambda$ exactly $\tau_i$ times.
	Now intuitively, from $\tau_1, \dots, \tau_\gamma$ we are able to recover the tuple computed by $P$ at the end of the run $\rho$ using existential Presburger arithmetic.
	So, in the sequel, $(i)$ we describe how to construct $\xi$ which defines $\P(L(A_\lambda))$ and $(ii)$ from $\xi$ we define the existential Presburger formula $\varphi(\ell)$ which holds iff there exists an accepting run of $P$ of length $\ell$.
	
	$(i)$
	From Theorem~\ref{thm:phdLin} applied on the $\NFA$ $A_\lambda$, there exist $m$ linear sets $L_i = \{  \vec{b}_i  + \sum_{j=1}^\gamma x_{i, j}\vec{p}_{i, j} \mid x_{i, j} \in \N\}$ such that $\P(L(A_\lambda)) = \bigcup_{i=1}^m L_i$.
	The linear set $L_i$ can be denoted by the following existential Presburger formula:
	\begin{equation} \label{eq:parikh_image}
	\xi_i(\vec{\tau}) =
	\exists \vec{x} \left[
		\bigwedge_{k=1}^\gamma
		\proj_k(\vec{\tau}) = \proj_k(\vec{b}_i) + \sum_{j=1}^\gamma \proj_j(\vec{x}) \times \proj_k(\vec{p}_{i, j})
	\right]
	\end{equation}
	Note that, $\vec{b}_i$ and $\vec{p}_{i, j}$ depends on $P$ only and can be computed in time $2^{\O(\gamma^2\log(\gamma n))}$.
	Also, Theorem~\ref{thm:phdLin} ensures that $m$ is at most $n^{\gamma^2+3\gamma+3}\gamma^{4\gamma+6}$.
	Moreover, for all $1 \leq i \leq m$, all $1 \leq j \leq \gamma$, we have that $\proj_{k}(\vec{b}_i) \leq n^{3\gamma+3}\gamma^{4\gamma+6}$ and $\vec{p}_{i, j} \in \{0,\dots,n\}^\gamma$.
	Since constants are encoded in binary, we have that $|\xi_i|$ is polynomial in $\gamma$ and logarithmic in $n$.
	
	$(ii)$
	Now, we explain how $\xi$ and the acceptance constraint\footnote{Lemma~\ref{lem:one-way_to_formula} $\psi$ can be trivially be generalised with an acceptance constraint which belongs to $\PA(\alpha, \beta)$ for some $\alpha, \beta \in \N$} of $P$ are glued together.
	Recall that $\Lambda = \{ \vec{a}_1, \dots, \vec{a}_\gamma \}$ where $\vec{a}_i \in \Z^d$ is the set of weight vectors of the original Parikh automaton $P$.
	The value of each dimension $1 \leq k \leq d$ at the end of a run can be computed from the number of visits $\tau_1, \dots, \tau_\gamma$ by $c_k = \sum_{j=1}^\gamma \tau_j \times \proj_k(\vec{a}_j)$ and the number of transition taken is $\ell = \sum_{j=1}^{\gamma} \tau_j$.
	Then, we define the formula $\varphi_i$ as follows.
	\begin{equation}
	\varphi_i(\ell) = 
	\exists \vec{\tau}, \exists \vec{c}
	\left[ \bigwedge
	\begin{cases}
		\xi_i(\vec{\tau}) \land \psi(\vec{c}) \land \ell = \sum_{j=1}^{\gamma} \tau_j
	\smallskip\\
		\bigwedge_{k=1}^d \proj_k(\vec{c}) = \sum_{j=1}^\gamma \proj_j(\vec{\tau}) \times \proj_k(\vec{a}_j)
	\end{cases}
	\right]
	\end{equation}
	We have that $|\varphi_i| = \O(|\xi_i|+|\psi|+ \gamma + d\gamma\log_2(\mu))$ where $\mu$ is the maximal absolute value appearing on weight vectors of $P$ i.e.\ $\mu = \max \{ \lVert \vec{a}_i \rVert \mid 1 \leq i \leq \gamma \}$.
	Recall that $|P| = \O(n + |\psi| + (d \log_2(\mu+1) + n^2)\times \gamma)$.
	Thus, $|\varphi_i|$ is polynomial in $|P|$, logarithmic in $n$ and can be computed in time $2^{\O(\gamma^2\log(\gamma n))}$.
\end{proof}

\section{Section~\ref{sec:closure}: Closure properties and comparison problems}

\begin{proof}[Proof of Theorem~\ref{thm:membership}]
	Given a $\twoNPA$ $P = (A, \lambda, \psi)$ with $A = (Q, Q_I, Q_F, \Delta)$ and a word $w \in \Sigma^*$, we construct an $\NPA$ $P_w$ such that $w \in L(P)$ iff $L(P_w) \neq \varnothing$.
	Intuitively, each state of $P_w$ encodes a configuration that appears in a run of $P$ on input $w$.
	We define $P_w = (A_w, \lambda', \psi)$ with $A_w = (Q', Q_I', Q_F', \Delta')$ where $Q' = \{(w_1, q, w_2) \mid q \in Q, w_1 w_2 =\lword w \rword\}$, and $Q_I' = \{(\varepsilon, q_0, \lword w \rword) \mid q_0 \in Q\}$, $Q_F' = \{(w, q_f, \varepsilon) \mid q_f \in F\}$ and $\lambda'$ is defined as following partial function for which $\Delta'$ is the domain:
	$$
		\bigcup \begin{cases}
			\left\{ \big((w_1, q_1, a w_2),a, (w_1 a, q_2, w_2)\big) \mapsto \vec{v} \mid
				(q_1, a, q_2) \in \Delta \land \lambda(q_1, a, q_2) = \vec{v} \land q_1 \in Q^\onright
			\right\}
		\smallskip\\
			\left\{ \big((w_1 a, q_1, w_2),a,(w_1, q_2, a w_2)\big) \mapsto \vec{v} \mid
				(q_1, a, q_2) \in \Delta \land \lambda(q_1, a, q_2) = \vec{v} \land q_1 \in Q^\onleft
			\right\}
		\end{cases}
	$$
	Note that a run of $P_w$ is a (one-way) $\NPA$ which simulates the sequence of configurations corresponding to a run of $P$ on input $w$, hence we have $L(P_w) \neq \varnothing$ iff $w\in L(P)$. 
	Non-emptiness and membership are shown to be \NPTimeC for $\NPA$ in~\cite{DBLP:conf/lics/FigueiraL15} which yields the statement.
\end{proof}

\end{document}